\documentclass[10pt,journal,comsoc]{IEEEtran}

\usepackage{cite}

\usepackage[T1]{fontenc}% optional T1 font encoding
\usepackage{bm}
\usepackage{amsmath}
\usepackage{etoolbox}

\usepackage[table,xcdraw]{xcolor}

\usepackage{amsthm}
\usepackage{float}
\usepackage{setspace}
\usepackage{amssymb}
\usepackage{stfloats}
\usepackage{cite}
\usepackage{ragged2e}
\usepackage[ruled,vlined,linesnumbered]{algorithm2e}
\usepackage{amsfonts}
\usepackage{mathrsfs}
\usepackage{amsmath,amsthm}
\usepackage{array,booktabs}
\usepackage{subfigure}
\usepackage{multirow}
\usepackage{cuted}
\usepackage{multicol}
\usepackage{graphicx}
\usepackage{subfigure}
\usepackage{graphicx,xcolor,bm}
\usepackage{hyperref}
\makeatother
\usepackage{threeparttable}
\usepackage{dcolumn}
\usepackage{setspace}
\usepackage{makecell}
\usepackage{lipsum}
\usepackage{enumerate}
\usepackage{mathrsfs}
\usepackage{bbm}

\newtheorem{Defn}{Definition}

\newtheorem{thm}{Theorem}

\usepackage{xfrac}

\usepackage{cite,bm}
\graphicspath{{figures/}}
\def\BibTeX{{\rm B\kern-.05em{\sc i\kern-.025em b}\kern-.08em
		T\kern-.1667em\lower.7ex\hbox{E}\kern-.125emX}}
\usepackage{balance}
\usepackage{nicefrac}
\allowdisplaybreaks
\begin{document}
	
	\title{STEPS: Semantic Contract-Guided Scheduling for LLM-Assisted Natural Language-Driven Edge AI Services}
	
	\author{Houyi Qi, \IEEEmembership{Graduate Student Member}, \IEEEmembership{IEEE}, Minghui Liwang, \IEEEmembership{Senior Member}, \IEEEmembership{IEEE}, Xianbin Wang, \IEEEmembership{Fellow}, \IEEEmembership{IEEE}, Xinlei Yi, \IEEEmembership{Senior Member}, \IEEEmembership{IEEE}, and Seyyedali Hosseinalipour, \IEEEmembership{Senior Member}, \IEEEmembership{IEEE}
		\thanks{
			
			This work was supported in part by National Natural Science Foundation of China under Grant nos. 62271424, 62503365; Shanghai Pujiang Programme under Grant no. 24PJD117; Shanghai Municipal Science and Technology Major Project under Grant no. 2021SHZDZX0100; Guizhou Provincial Science and Technology Plan Project ([2026]209); Aeronautical Science Foundation of China under Grant no. 2023Z066038001; Chinese Academy of Engineering, Strategic Research and Consulting Program under Grant no. 2023-XZ-65; the U.S. National Science Foundation (NSF) under Grant no. ECCS-2512911. H. Qi (houyiqi@tongji.edu.cn), M. Liwang (minghuiliwang@
			tongji.edu.cn), and X. Yi (xinleiyi@tongji.edu.cn) are with the Shanghai Research Institute for Intelligent Autonomous Systems, State Key Laboratory of Autonomous
			Intelligent Unmanned Systems, Frontiers Science Center for Intelligent
			Autonomous Systems, and Department of Control Science and Engineering, Tongji University, Shanghai, China. X. Wang (xianbin.wang@uwo.ca) is with the Department of Electrical and Computer Engineering, Western University, Ontario, Canada. S. Hosseinalipour (alipour@buffalo.edu) is with the Department of Electrical Engineering, University at Buffalo-SUNY, USA.}}
	
	%\setlength{\abovedisplayskip}{2pt}
	%\setlength{\belowdisplayskip}{2pt}
	%\setlength{\skip\footins}{6pt}
	%\setlength{\footnotesep}{0pc}
	%\setstretch{.94}    % adjust the number as needed
	%\everymath{\small}
	%\everydisplay{\footnotesize}

	\IEEEtitleabstractindextext{
		
		\begin{abstract}
			\justifying
			Edge user/service scheduling has become a cornerstone of distributed AI systems, governing \textit{where} and \textit{how} AI services are executed under limited communication and computing resources. Nevertheless, existing edge scheduling frameworks typically assume that service requirements are explicitly specified as numerical quantities/constraints, such as latency bounds or energy budgets. In practice, however, users often express their service expectations through ambiguous and context-dependent natural language descriptions (e.g., “finish quickly” or “save energy”), creating a fundamental gap between user intent and scheduling decisions. Bridging this semantic-to-optimization gap is particularly challenging in dynamic edge environments with time-varying resources and execution conditions. To address this challenge, we propose \textit{\underline{s}emantic con\underline{t}ract--guided \underline{e}dge \underline{p}otential \underline{s}cheduling (STEPS)}, a natural language-driven scheduling framework that introduces \textit{semantic contracts} as executable interfaces between user-side semantics and edge-side decision making. {\color{black}Specifically, an LLM-assisted semantic parser extracts service requirements and raw confidence scores from natural language requests. The extracted information is then converted into preference weights, fulfillment bounds, and semantic uncertainty through confidence calibration and interpretable mapping rules.} {\color{black}Based on these contracts, STEPS formulates intended per-slot edge scheduling as an exact potential game for execution-node selection, computing-resource provisioning, and bandwidth allocation, followed by capacity-feasible resource realization. STEPS further incorporates \textit{fulfillment-driven feedback signals} to enable adaptive scheduling under evolving service and network conditions. We establish the existence of pure-strategy Nash equilibrium, prove finite-step convergence of uncapped strict best-response dynamics, and characterize the boundedness of the projected adaptive-control parameters. Experimental results under synthetic and Melbourne-based hybrid settings demonstrate improved contract fulfillment and reduced service loss relative to several adapted benchmarks. Post-drift evaluations further show that feedback adaptation improves both fulfillment and admission coverage under the tested conditions.} 
		\end{abstract}
		
		\begin{IEEEkeywords}
			Natural Language-Driven Scheduling, Semantic Contract, Edge Intelligence, Potential Game.
		\end{IEEEkeywords}
	}

	\maketitle
	\IEEEdisplaynontitleabstractindextext
	
	\IEEEpeerreviewmaketitle
	\section{Introduction}
	\IEEEPARstart{R}{ecent} advances in artificial intelligence (AI) and computing technologies have led to the emergence of \textit{networked AI}, where AI models are deployed to enable intelligent services, such as industrial automation, autonomous driving, and immersive human-machine interaction~\cite{mao2024green,he2025task}. Unlike conventional cloud-based AI services, networked AI operates over the distributed and time-varying device-edge-cloud continuum, where communication (e.g., data transmission) and computation (e.g., modeling training and inference) are tightly coupled with execution feedback loops. In this context, edge computing provides a natural substrate of \underline{e}dge services for \underline{net}worked \underline{AI} (E4NetAI), enabling latency-sensitive tasks to be executed closer to end users and thus reducing end-to-end service latency~\cite{jia2025comprehensive,wang2024cooperative}. Motivated by this premise, a major body of work on E4NetAI has focused on distributed device/service scheduling and network orchestration schemes~\cite{mao2024green,qi2025future}. However, such works largely follow a \textit{parameter-driven} paradigm that requires users to explicitly specify low-level quantitative requirements (e.g., delay tolerances, energy budgets, and payment limits)~\cite{qi2025future,fan2025latency}. This assumption is often unrealistic for non-expert users who lack backgrounds in the detailed pipelines of innovative technologies. In fact, 
	% instead of calibrating maximum tolerable latency, energy budgets, payment limits, or preference weight vectors, 
	ordinary users mostly tend to express their subjective preferences and service expectations in natural language~\cite{mekrache2024llm,jacobs2021hey,ji2024computational}, e.g., ``finish this task as quickly as possible without incurring excessive cost''. As a result, the user interactions in E4NetAI services are expected to shift from \textit{expert-oriented parameter configuration} toward \textit{human-centric semantic interaction}. While this shift lowers the access barrier to networked AI services for the users, it also poses a major challenge to the pipeline of networked services: \textit{transforming  implicit and uncertain natural language-based user intents into executable and optimizable service specifications that are consistent with edge resource constraints.}
	
	To address this challenge, recent advances in intent-based networking (IBN) and large language models (LLMs) provide new opportunities. Specifically, IBN translates high-level user/operator intents into network policies and configurations~\cite{wang2025survey,zhang2025communication,zhang2025split}, while LLM-assisted network orchestration enables natural language interpretation and automated service configuration~\cite{mekrache2024llm,qin2025generative}. Nevertheless, directly applying existing IBN and LLM-assisted orchestration techniques to E4NetAI remains insufficient. This is because while IBN and LLMs are effective at interpreting high-level user/operator intents and generating service policies, they do not inherently provide an optimization-compatible representation that can reliably bridge ambiguous user semantics and resource-constrained edge decision-making. 
	Motivated by this shortcoming, we develop \textit{a unified semantic contract framework that bridges intent interpretation and execution through contract-based decision evolution}. In particular, our framework aims to address  
	% Realizing such a framework requires addressing several fundamental challenges spanning semantic grounding, contract-guided scheduling, and fulfillment-driven adaptation. Accordingly, we distill
	the following three research questions (RQs).
	
	\noindent
	$\bullet$ \textit{RQ 1: How can unstructured natural language service intents be systematically grounded into a structured and constraint-consistent decision representation for E4NetAI?} Addressing this question is important because natural language service requests are inherently ambiguous, qualitative, and context-dependent, whereas edge schedulers require precise representations that can be evaluated and optimized under resource constraints. Moreover, while existing IBN and LLM-assisted orchestration approaches can translate user intents into high-level policies or control directives~\cite{wang2025survey,mekrache2024llm}, they still lack an explicit intermediate representation that simultaneously preserves semantic meaning and enables optimization-compatible decision making. Therefore, a key challenge is to establish an executable interface between user-side semantics and edge-side scheduling decisions while maintaining feasibility under system performance, cost, and execution constraints.
	
	\noindent
	$\bullet$ \textit{RQ 2: How can semantic contracts reshape the structure of edge scheduling under tightly coupled constraints and semantic uncertainty?}
	This question stems from the fact that, once natural language intents are transformed into semantic contracts, service requirements are no longer represented solely by fixed numerical constraints. Instead, they become preference-aware, uncertainty-aware, and fulfillment-oriented objectives that must be simultaneously considered during user/service scheduling. Consequently, conventional edge scheduling formulations, which primarily optimize resource utilization or quality-of-service metrics under predefined constraints~\cite{qi2025future,li2024optimal}, are no longer applicable. Therefore, a key challenge is to design a scheduling framework that jointly determines execution placement, computing-resource provisioning, and bandwidth allocation while accounting for semantic contract requirements, resource limitations, and semantic uncertainty.
	
	\noindent
	$\bullet$ \textit{RQ 3: How can post-execution contract-satisfaction feedback enable principled adaptation in non-stationary E4NetAI systems?}
	This question arises because E4NetAI operates in highly dynamic environments, where user preferences,  resource availability, and execution conditions  evolve over time. Nevertheless, existing adaptive optimization and drift-aware learning approaches primarily focus on observable variations in data distributions, network states, or model performance~\cite{kalntis2024adaptive}. As a result, they largely overlook the mismatch between semantic contract expectations and realized service outcomes, which can only be evaluated after task execution. Subsequently, a key challenge is to construct meaningful contract-satisfaction feedback signals and distinguish between changes originating from user-side semantic evolution (i.e., \textit{semantic-request drift}, where the expectations expressed in user requests change over time) and changes caused by system-side execution conditions (i.e., \textit{contract-fulfillment drift}, where contract-satisfaction outcomes vary due to fluctuations in resources, network conditions, or service performance).
	
	To answer the above RQs, we propose \textit{\underline{s}emantic con\underline{t}ract--guided \underline{e}dge \underline{p}otential \underline{s}cheduling (STEPS)}. In a nutshell, \textit{for RQ1,} STEPS introduces LLM-driven semantic contracts as an executable interface between natural language intents and edge service decisions.
	% In doing so, STEPS leverages an LLM-assisted semantic parser to interpret user requests and extract semantic service requirements together with a confidence score that quantifies the reliability of the semantic interpretation, which are then converted into service preferences, fulfillment bounds, and semantic uncertainty, thereby enabling optimization-compatible decision making. 
	\textit{To address RQ 2}, STEPS formulates a semantic contract-guided edge scheduling problem under coupled communication-computation constraints and develops a potential-game approach for pure-strategy Nash equilibrium computation. \textit{To handle RQ 3}, STEPS evaluates post-execution contract fulfillment and constructs feedback signals to enable adaptive
	scheduling under evolving service and network conditions. Subsequently, our main contributions can be summarized as follows.
	
	\noindent $\bullet$ 
	We propose STEPS, which 
	% is one of the first semantic contract-based adaptive scheduling framework for natural language-driven E4NetAI.
	% STEPS 
	establishes one of the first closed-loop pipelines from natural language service requests to semantic abstraction, contract-guided edge scheduling, post-execution contract-satisfaction evaluation, and fulfillment-driven adaptation. Further, through introducing the semantic contract as an executable  interface, STEPS bridges user-side natural language semantics and edge resource optimization.
	
	\noindent $\bullet$ 
	We design an LLM-empowered semantic contract generation engine
	% that converts natural language requests into optimization-compatible service contracts.
	% Specifically, an LLM-assisted semantic parser
	that
	analyzes each request and extracts semantic service requirements (e.g., delay and cost sensitivity) together with a confidence score that quantifies the reliability of the semantic interpretation. Based on these extracted semantics, we utilize interpretable mapping rules to combine user requirements, task attributes, and system-side context to construct a semantic contract,
	% consisting of service preferences, fulfillment bounds, and semantic uncertainty. 
	which serves as an executable representation of user intent that can be directly used in scheduling decisions.
	
	% By restricting the LLM to semantic interpretation rather than resource optimization, our framework avoids relying on unconstrained LLM-generated resource allocations while preserving interpretability and optimization consistency.
	
	\noindent $\bullet$ 
	% We develop a contract-guided edge scheduling engine for distributed execution-node selection and resource allocation in E4NetAI. 
	Guided by the semantic contracts generated from user requests, we develop an edge scheduling engine that jointly determines task-execution placement, computing-resource provisioning, and bandwidth allocation while respecting communication and computation resource constraints. 
	% The scheduling objective explicitly captures the tradeoff among service delay, energy consumption, monetary cost, trustworthiness-for-execution, and semantic contract fulfillment. 
	To efficiently address the scheduling problem, we formulate a per-slot penalized surrogate as an exact potential game and develop an asynchronous best-response mechanism coordinated by the edge service platform for pure-strategy Nash equilibrium computation.
	% This design enables scalable and contract-aware scheduling without requiring centralized optimization across all users and edge servers.
	
	\noindent $\bullet$ 
	% We design a set of feedback signals, which
	%%  for non-stationary E4NetAI. Rather than adapting solely to changes in data distributions or network conditions, our proposed engine 
	% evaluate how well the executed services satisfy their semantic contracts. Specifically, they capture the semantic-request drift and contract-fulfillment drift, and instantaneous fulfillment pressure arising from contract violations, while using admission pressure to regulate the acceptance of highly uncertain requests. Based on these feedback signals, our framework dynamically updates semantic-admission thresholds, contract conservativeness factors, and edge-coordination gains. 
	We design a set of feedback signals that evaluate how well the executed services satisfy their semantic contracts. Specifically, they capture semantic-request drift, contract-fulfillment drift, and fulfillment pressure arising from contract violations, while using admission pressure to regulate highly uncertain requests. Based on these signals, our framework dynamically updates semantic-admission thresholds, contract conservativeness factors, and edge-coordination gains.

	% This design enables fulfillment-driven adaptation and sustained scheduling performance under evolving user semantics, resource availability, and service conditions.
	
	\noindent $\bullet$ 
	% We provide theoretical analysis and numerical evaluations of STEPS.
	% On the theoretical side, 
	We characterize the exact potential game structure of the contract-guided scheduling problem and establish key properties of the proposed framework, including the existence of pure-strategy Nash equilibrium, finite-step convergence of the asynchronous best-response dynamics, and boundedness of the feedback-driven adaptive update process.

	\noindent $\bullet$  Through numerical evaluations, we demonstrate that STEPS improves semantic contract fulfillment, reduces contract-guided service loss, and maintains robust adaptation under ambiguous natural language requests, semantic uncertainty, and non-stationary network conditions.

	\section{Literature Review}
	
	In the following, we review the literature across three interrelated domains
	% : (i) IBN and LLM-assisted service orchestration, (ii) edge AI service scheduling and provisioning, and (iii) drift-aware adaptive optimization. Throughout the review, we emphasize 
	and identify the fundamental differences between these research directions and STEPS. 
	% Table~\ref{tab:related_work_summary} summarizes our discussions.

	~\textit{(i) IBN and LLM-Assisted Service Orchestration:}
	IBN serves as an important paradigm for simplifying network management by allowing users/operators to express high-level service intents instead of manually specifying low-level configurations~\cite{wang2025survey,mekrache2024intent}. 
	% Within this paradigm, existing research has primarily focused on translating user intents into network policies, service descriptors, and management actions through processes such as intent decomposition, policy generation, and closed-loop assurance. 
	Recent advances in LLMs have further expanded the capabilities of intent-driven orchestration by enabling natural language understanding and automated decision support. For example, \textit{Mekrache et al.}~\cite{mekrache2024intent} proposed an LLM-centric intent-based management architecture for next-generation networks that supports intent decomposition, translation, activation, and assurance. \textit{Mekrache and Ksentini}~\cite{mekrache2024llm} further studied LLM-enabled intent-driven service configuration, where natural language intents are translated into network service descriptors. At a larger operational scale, \textit{Wang et al.}~\cite{wang2025intent} developed a multi-agent LLM framework that decomposes complex network management workflows into coordinated subtasks executed by multiple agents. Beyond traditional network management, recent studies have extended intent-driven paradigms toward computing and edge-service environments. For example, \textit{Akbari et al.}~\cite{akbari2025intentcontinuum} investigated LLM-assisted intent-based computing, where user-defined intents are monitored and handled through diagnosis and reconfiguration. \textit{Qin et al.}~\cite{qin2025generative} explored the integration of generative AI and intent-driven wireless communications to facilitate interactions between users and communication systems. \textit{Sun et al.}~\cite{sun2026igaa} proposed an agentic-AI  framework that maps user intents into resource-oriented representations and leverages generative meta-learning to improve cross-scenario generalization. \textit{These prior studies
		% have explored LLM-enabled and intent-driven network orchestration, they 
		predominantly treat natural language inputs as mechanisms for generating policies, management actions, or scheduling directives.} As a result, they provide limited support for representing user requests as uncertainty-aware semantic requirements that can be directly incorporated into resource-constrained optimization. Subsequently, our work introduces semantic contracts as executable interfaces between user-side semantics and edge-side scheduling decisions, enabling interpretable semantic grounding, contract-guided optimization, and fulfillment-driven adaptation within a unified framework.
	
	~\textit{(ii) Edge AI Service Scheduling and Provisioning:}
	% Edge intelligence has been widely studied to support latency-sensitive and computation-intensive AI services by bringing computation closer to end users. 
	% With the rapid development of generative AI and foundation models, 
	Edge networks are increasingly utilized for delivering real-time and personalized AI services to the users under limited communication/computing resources~\cite{mao2024green,xu2024unleashing,qu2025mobile,Abdisarabshali2026Hierarchical}.  Consequently, a significant body of research has focused on how to efficiently schedule and provision AI workloads across devices, edge servers, and cloud resources. 
	% In particular, to improve service efficiency under resource constraints, existing studies have explored a variety of edge AI scheduling and orchestration mechanisms. 
	For example, \textit{Li and Bi}~\cite{li2024optimal} jointly optimized AI model partitioning and wireless resource allocation for device-edge collaborative inference. \textit{Xiao et al.}~\cite{xiao2024adaptive} studied a content-aware compression and offloading framework for efficient edge-based vision inference. \textit{Zhang et al.}~\cite{zhang2025beyond} investigated resource-efficient deployment of generative LLM inference at the network edge through  optimizing quantization and resource scheduling.
	% under heterogeneous latency and accuracy requirements. 
	\textit{These prior studies generally assume that service requirements are explicitly stated and available in the form of  numerical specifications (e.g., latency bounds or accuracy targets); however, user requests are often expressed through qualitative and ambiguous descriptions.} To this end, our work  addresses the semantic-to-optimization gap by transforming natural language user requests into semantic contracts that directly guide resource-constrained service provisioning at the network edge.

	~\textit{(iii) Drift-Aware and Adaptive Optimization:}
	Adaptive optimization and drift-aware adaptation have attracted growing attention for detecting changes in system behavior and adapting learning models or control policies to maintain performance over time. For example, \textit{Ganguly and Aggarwal}~\cite{ganguly2024online} developed an online federated learning framework integrating drift detection to improve performance under online data. \textit{Gudepu et al.}~\cite{gudepu2024drift} proposed a drift-management framework for Open Radio Access Network (O-RAN), that incorporates drift detection and adaptive mitigation mechanisms to reduce service-level-agreement violations and improve resource utilization. \textit{Kalntis et al.}~\cite{kalntis2024adaptive} studied adaptive service provisioning for virtualized base stations in O-RAN, where resource-control decisions are updated according to time-varying network conditions/states. \textit{Uzlaner et al.}~\cite{uzlaner2025asynchronous} developed a modular drift-detection framework that determines when and which components of a deep learning-based receiver should be retrained under dynamic channel conditions. \textit{Ameur et al.}~\cite{ameur2025dual} designed a dual self-attention mechanism for detecting data drift, label drift, and concept drift in 6G networks. \textit{Although these studies enhance robustness and adaptability in non-stationary environments, their considered adaptation signals are predominantly derived from observable system characteristics, such as data distributions, model performance, and network condition.} However, natural language-driven E4NetAI introduces an additional source of non-stationarity arising from evolving user semantics and their resulting fulfillment outcomes. Consequently, the impact of semantic-request drift (i.e., changes in the characteristics of user requests) and contract-fulfillment drift (i.e., changes in contract-satisfaction outcomes caused by variations in execution conditions) remains largely unexplored. Our work addresses this gap by jointly modeling semantic-request drift, fulfillment drift, and instantaneous fulfillment pressure, and leveraging these signals to drive fulfillment-aware edge scheduling adaptation mechanisms.

	\section{System Model and Problem Formulation}
	\label{sec:system_problem}
	\label{sec:system_model}
	\label{sec:overview}
	
	\begin{figure}[t]
		\centering
		\setlength{\abovecaptionskip}{-1mm}
		\includegraphics[width=1\columnwidth]{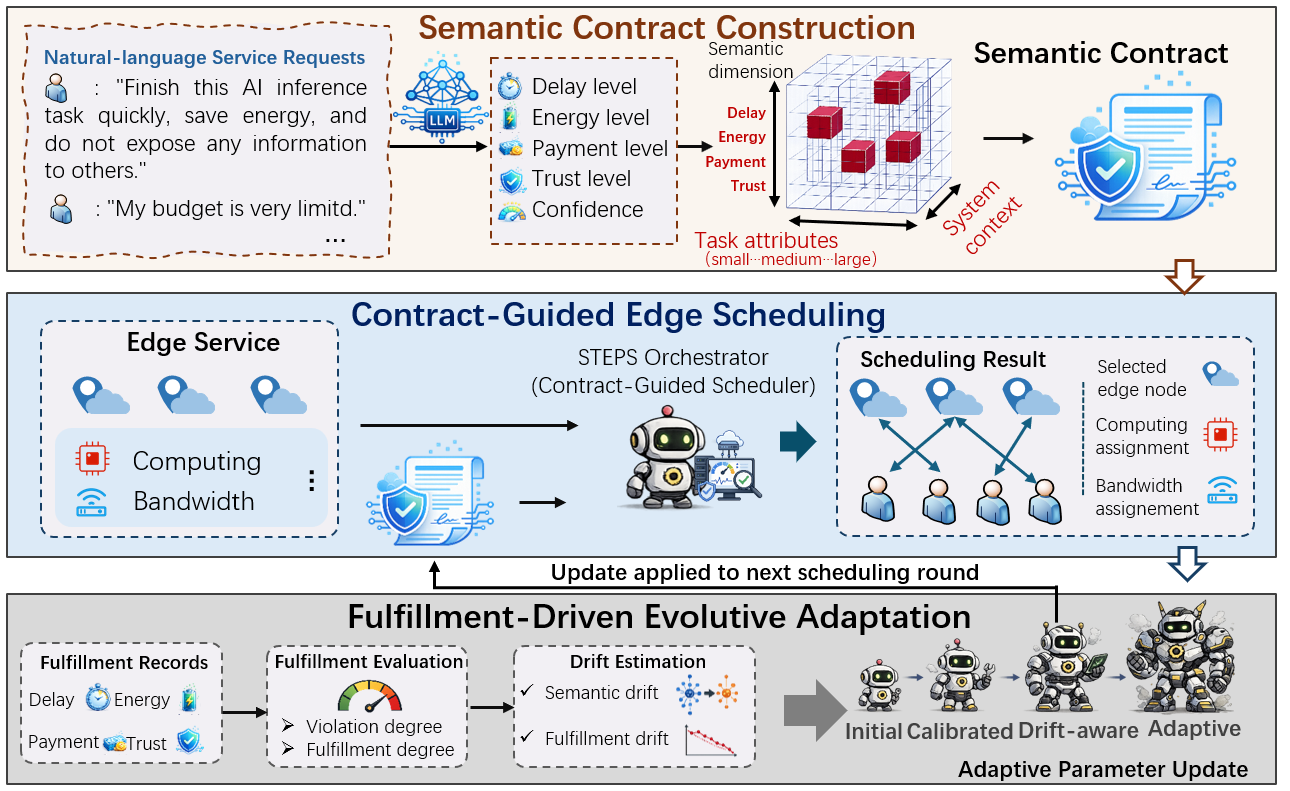}
		\caption{STEPS framework for natural language-driven E4NetAI. Natural-language service requests are first transformed into semantic contracts, which then guide edge scheduling. Execution feedback further drives fulfillment-aware adaptive parameter updates for subsequent scheduling rounds.}
		\label{fig:steps_framework}
	\end{figure}
	
	% As illustrated in Fig.~\ref{fig:steps_framework}, STEPS operates over a time-varying E4NetAI system across a time-slotted horizon $\mathcal{T}=\{1,2,\ldots,T\}$. At each timeslot $t\in \mathcal{T}$, the edge service platform (ESP) observes the system context and receives user device (UD) requests characterized by task attributes and natural language service descriptions. The descriptions are translated/parsed into semantic service requirements with associated confidence and subsequently instantiated as semantic contracts that encode user preferences, fulfillment bounds, and semantic uncertainty. The resulting contracts are then calibrated to account for uncertainty and integrated into edge execution decisions. Based on this semantic contract framework, 
	In the following, we present the system model and formulate the long-term contract-guided scheduling problem of STEPS.

	\subsection{Time-Slotted Edge Service Platform}
	\label{sec:edge_system}
	
	As illustrated in Fig.~\ref{fig:steps_framework}, STEPS operates over a time-varying E4NetAI system across a time-slotted horizon {\small$\mathcal{T}=\{1,2,\ldots,T\}$}. At each timeslot {\small$t\in \mathcal{T}$}, the edge service platform (ESP) observes the system context and receives user device (UD) requests characterized by task attributes and natural language service descriptions.
	We consider an ESP formed by a set of geographically distributed edge servers (ESs), denoted by 
	{\small$\mathcal{E}=\{e_1,e_2,\ldots,e_{|\mathcal{E}|}\}$}, which collaboratively deliver E4NetAI services. To account for network dynamics, at each timeslot {\small$t\in\mathcal{T}$}, we use {\small$\mathcal{U}^{(t)}$} to denote the set of UDs that submit service requests to the ESP. 
	% After semantic contract generation and admission, 
	The ESP ultimately determines whether the task of each UD should be executed \textit{locally} or \textit{offloaded} to an ES for remote execution.
	
	To enable a unified treatment of local and edge execution decisions, we abstract all possible task execution locations as execution nodes (ENs). Specifically, we define the EN set as
	{\small$\mathcal{E}_0=\{e_0\}\cup\mathcal{E}$}, 
	where the virtual node {\small$e_0$} denotes local execution on the requesting UD and each {\small$e_j\in\mathcal{E}$} corresponds to an ES. For each ES {\small$e_j$}, we let {\small$F_{j}^{(t)}$} and {\small$B_{j}^{(t)}$} denote its available computing and bandwidth capacity at timeslot {\small$t$}, respectively. For local execution, {\small$F_{i}^{\mathsf{loc},(t)}$} denotes the available computing capacity of UD {\small$u_i$}. Furthermore, each EN is associated with a trustworthiness-for-execution (T4E) score. Specifically, {\small$\sigma_{j}^{\mathsf{edge},(t)}\in[0,1]$} denotes the T4E score of ES {\small$e_j$}, while {\small$\sigma_{i,0}^{\mathsf{loc},(t)}\in[0,1]$} denotes the T4E score of UD {\small$u_i$}'s local device. Larger T4E scores indicate a higher degree of confidence that the corresponding EN can reliably execute the task and satisfy the trustworthiness-related requirement specified in the semantic contract\footnote{In this paper, trust is not modeled as a standalone security construct; instead, it is incorporated as one dimension of semantic contract fulfillment during task execution.}.
	The above resource and trust states, together with channel, queueing, and price information, constitute the \textit{system-side context} used by the ESP before contract generation and scheduling. Specifically, the system-side context at timeslot $t$ is defined as
	\begin{equation}
		\hspace{-2mm}\resizebox{0.94\linewidth}{!}{$\begin{aligned}
				\xi^{(t)}
				\hspace{-.5mm}\triangleq\hspace{-.5mm}
				\big(
				\mathbf{F}^{(t)},\mathbf{B}^{(t)},\mathbf{F}^{\mathsf{loc},(t)},
				\mathbf{H}^{(t)},\mathbf{Q}^{(t)},
				\boldsymbol{\lambda}^{\mathsf{f},(t)},\boldsymbol{\lambda}^{\mathsf{b},(t)},
				\boldsymbol{\sigma}^{\mathsf{edge},(t)},\boldsymbol{\sigma}^{\mathsf{loc},(t)}
				\big),
			\end{aligned}$}\hspace{-2mm}
		\label{eq:system_context}
	\end{equation}
	where {\small $\mathbf{F}^{(t)}=[F_j^{(t)}]_{e_j\in\mathcal{E}}$} and {\small $\mathbf{B}^{(t)}=[B_j^{(t)}]_{e_j\in\mathcal{E}}$} denote the available computing and bandwidth capacities of ESs, {\small $\mathbf{F}^{\mathsf{loc},(t)}=[F_i^{\mathsf{loc},(t)}]_{u_i\in\mathcal{U}^{(t)}}$} denotes the available local computing capacities of UDs, {\small$\mathbf{H}^{(t)}=[h_{i,j}^{(t)}]_{u_i\in\mathcal{U}^{(t)},e_j\in\mathcal{E}}$} denotes the UD-to-ES wireless channel-gain matrix, and {\small$\mathbf{Q}^{(t)}=[Q_j^{(t)}]_{e_j\in\mathcal{E}}$} denotes the pre-decision queueing-delay state of ESs. Also, {\small$\boldsymbol{\lambda}^{\mathsf{f},(t)}=[\lambda_j^{\mathsf{f},(t)}]_{e_j\in\mathcal{E}}$} and {\small$\boldsymbol{\lambda}^{\mathsf{b},(t)}=[\lambda_j^{\mathsf{b},(t)}]_{e_j\in\mathcal{E}}$} denote the current unit price coefficients of ES computing and bandwidth resources, respectively. Moreover, the vectors {\small$\boldsymbol{\sigma}^{\mathsf{edge},(t)}=[\sigma_j^{\mathsf{edge},(t)}]_{e_j\in\mathcal{E}}$ and $\boldsymbol{\sigma}^{\mathsf{loc},(t)}=[\sigma_{i,0}^{\mathsf{loc},(t)}]_{u_i\in\mathcal{U}^{(t)}}$} collect the T4E scores of ESs and local UD devices, respectively. This system-side context is maintained by the ESP and will be used in subsequent semantic contract generation and contract-guided scheduling.

	\subsection{Natural Language Request and Semantic Contract}
	\label{sec:nl_request_model}
	
	At each timeslot {\small$t$}, each UD {\small$u_i\in\mathcal{U}^{(t)}$} submits a service request consisting of a task attribute vector {\small$q_i^{(t)}$} and its associated service-expectation description (SED) {\small$\mathcal{L}_i^{(t)}$}, represented by {\small$(q_i^{(t)},\mathcal{L}_i^{(t)})$}. The task attribute vector is defined as {\small$q_{i}^{(t)}=(d_{i}^{(t)},c_{i}^{(t)})$}, where {\small$d_{i}^{(t)}$} denotes the input data size and {\small$c_{i}^{(t)}$} denotes the computation workload required for task completion. Together, these attributes characterize the task's fundamental communication and computation requirements. The SED {\small$\mathcal{L}_{i}^{(t)}$} expresses user service requirements at a semantic level in natural language, such as ``I do not have much time and budget'' or ``my phone is running out of power.'' Consequently, unlike conventional edge scheduling frameworks that rely on explicitly stated numerical constraints, STEPS treats natural language descriptions as the primary user-side input and translates them into operational requirements through a semantic-contract interface.
	{\color{black}
		Specifically, the SED $\mathcal{L}_{i}^{(t)}$ is processed by an LLM-assisted semantic parser that extracts four semantic service levels and a raw confidence score $c_i^{\mathsf{raw},(t)}$. The confidence score is subsequently calibrated using an independent calibration set to obtain $c_i^{\mathsf{conf},(t)}\in[0,1]$. The resulting semantic feature vector is $\mathbf{y}_{i}^{(t)}=(\ell_{i}^{\mathsf{D},(t)},\ell_{i}^{\mathsf{E},(t)},\ell_{i}^{\mathsf{M},(t)},\ell_{i}^{\mathsf{S},(t)},c_{i}^{\mathsf{conf},(t)})$, where $\ell_{i}^{\mathsf{D},(t)}$, $\ell_{i}^{\mathsf{E},(t)}$, $\ell_{i}^{\mathsf{M},(t)}$, and $\ell_{i}^{\mathsf{S},(t)}$ denote the delay-sensitivity level, energy-sensitivity level, monetary-cost sensitivity level, and T4E requirement level, respectively. These quantities are transformed into optimization-compatible contract parameters through interpretable mapping rules that incorporate task attributes and system-side context.
	}
	A key component of STEPS is the semantic contract, which serves as the executable interface between natural language service requirements and optimization-driven edge scheduling. Specifically, based on the semantic feature vector {\small$\mathbf{y}_{i}^{(t)}$}, task attributes {\small$q_{i}^{(t)}$}, and system-side context {\small$\xi^{(t)}$}, the ESP generates a semantic contract for UD {\small$u_i$} at timeslot {\small$t$} as {\small$\mathcal{C}_{i}^{(t)}=\mathcal{F}_{\mathsf{sc}}(\mathbf{y}_{i}^{(t)},q_{i}^{(t)};\xi^{(t)})=(\mathbf{w}_{i}^{(t)},\mathbf{g}_{i}^{(t)},\rho_{i}^{(t)})$}, where {\small$\mathcal{F}_{\mathsf{sc}}(\cdot)$} denotes the semantic contract generation function\footnote{The semicolon in {\scriptsize$\mathcal{F}_{\mathsf{sc}}(\mathbf{y}_{i}^{(t)},q_{i}^{(t)};\xi^{(t)})$} separates the UD-specific inputs from the system-side conditioning context. Specifically, {\small$\mathbf{y}_{i}^{(t)}$} and {\small$q_i^{(t)}$} describe the semantic and task-specific information of UD {\small$u_i$}, whereas {\small$\xi^{(t)}$} represents the system-side context that conditions the semantic contract generation.}. The semantic contract {\small$\mathcal{C}_{i}^{(t)}$} consists of three components: a preference vector {\small$\mathbf{w}_{i}^{(t)}$}, a fulfillment-bound vector {\small$\mathbf{g}_{i}^{(t)}$}, and a semantic-uncertainty measure {\small$\rho_{i}^{(t)}$}\footnote{The LLM-assisted parser extracts semantic levels and confidence scores, whereas continuous contract parameters are generated through interpretable mapping rules. This separates semantic interpretation from resource-allocation decisions.}.
	
	In particular, the preference vector {\small$\mathbf{w}_{i}^{(t)}=[w_{i}^{\mathsf{D},(t)},w_{i}^{\mathsf{E},(t)},w_{i}^{\mathsf{M},(t)},w_{i}^{\mathsf{S},(t)}]$} captures what the user prioritizes during service provisioning, where the four elements represent the relative importance of delay, energy consumption, monetary cost, and T4E, respectively. Let {\small$\hat{\ell}_{i}^{\mathsf{k},(t)}\in[0,1]$} denote the normalized semantic level of dimension {\small$\mathsf{k}\in\{\mathsf{D},\mathsf{E},\mathsf{M},\mathsf{S}\}$}. The preference weights are generated according to
	\begin{equation}
		w_{i}^{\mathsf{k},(t)}
		=
		\frac{\exp(\theta_w \hat{\ell}_{i}^{\mathsf{k},(t)})}
		{\sum_{\mathsf{r}\in\{\mathsf{D},\mathsf{E},\mathsf{M},\mathsf{S}\}}
			\exp(\theta_w \hat{\ell}_{i}^{\mathsf{r},(t)})},
		\quad
		\mathsf{k}\in\{\mathsf{D},\mathsf{E},\mathsf{M},\mathsf{S}\}.
		\label{eq:weight_mapping}
	\end{equation}
	The parameter {\small$\theta_w\ge0$} controls the degree of preference differentiation. By construction, the weights satisfy {\small$\sum_{\mathsf{k}\in\{\mathsf{D},\mathsf{E},\mathsf{M},\mathsf{S}\}}w_{i}^{\mathsf{k},(t)}=1$ and $w_{i}^{\mathsf{k},(t)}\ge0$}.
	Also, the fulfillment-bound vector {\small$\mathbf{g}_{i}^{(t)}=[D_{i}^{\mathsf{max},(t)},E_{i}^{\mathsf{max},(t)},M_{i}^{\mathsf{max},(t)},S_{i}^{\mathsf{min},(t)}]$} specifies the service requirements that should be satisfied during execution. Here, {\small$D_{i}^{\mathsf{max},(t)}$}, {\small$E_{i}^{\mathsf{max},(t)}$}, {\small$M_{i}^{\mathsf{max},(t)}$}, and {\small$S_{i}^{\mathsf{min},(t)}$} denote the maximum tolerable delay, energy consumption, monetary cost, and the minimum required T4E level, respectively.

	{\color{black}
		To generate these bounds, the ESP first computes reference service indicators under the current system context, denoted by $\bar{D}_{i}^{(t)}$, $\bar{E}_{i}^{(t)}$, and $\bar{M}_{i}^{(t)}$. Specifically, a fixed reference policy $\pi^{\mathsf{ref}}$ selects the nearest ES with a predefined computing and bandwidth configuration to estimate the service metrics defined in Sec.~\ref{sec:edge_execution_model}. These reference indicators are determined before contract-guided scheduling and do not depend on the resulting resource allocation.
	}Using them, the fulfillment bounds are:
	\begin{equation}
		\begin{aligned}
			\hspace{-2mm}
			&D_{i}^{\mathsf{max},(t)}
			=\bar{D}_{i}^{(t)}
			\delta_{\mathsf{D}}(\hat{\ell}_{i}^{\mathsf{D},(t)}),~~
			E_{i}^{\mathsf{max},(t)}
			=\bar{E}_{i}^{(t)}
			\delta_{\mathsf{E}}(\hat{\ell}_{i}^{\mathsf{E},(t)}),\\
			\hspace{-2mm}
			&M_{i}^{\mathsf{max},(t)}
			=\bar{M}_{i}^{(t)}
			\delta_{\mathsf{M}}(\hat{\ell}_{i}^{\mathsf{M},(t)}),~~
			S_{i}^{\mathsf{min},(t)}
			=[S_0+\delta_{\mathsf{S}}
			\hat{\ell}_{i}^{\mathsf{S},(t)}]_0^1.
		\end{aligned}
		\label{eq:bound_mapping}
	\end{equation}
	Here, {\small$[x]_0^1=\min\{\max\{x,0\},1\}$} denotes the projection of {\small$x$} onto {\small$[0,1]$}. The functions {\small$\delta_{\mathsf{D}}(\cdot)$}, {\small$\delta_{\mathsf{E}}(\cdot)$}, and {\small$\delta_{\mathsf{M}}(\cdot)$} are non-increasing mapping functions that translate semantic service requirements into quantitative fulfillment requirements. For example, a higher delay-sensitivity level leads to a smaller {\small$\delta_{\mathsf{D}}(\hat{\ell}_{i}^{\mathsf{D},(t)})$} and hence a tighter delay bound\footnote{In the experiments, these mapping functions use fixed lookup tables corresponding to low, medium, and high semantic levels, as summarized in Table~\ref{tab:dataset-contract}.}.
	
	{\color{black}
		Before evaluating semantic uncertainty, the ESP constructs an individually feasible action set $\mathcal{A}_{i}^{\mathsf{ind},(t)}$ for each incoming UD from the available local and edge resource packages. Each candidate action specifies an execution node and resource configuration satisfying the individual physical resource limits. This set is constructed before semantic admission, independently of contract calibration and the actions of other UDs, and is assumed nonempty for the considered requests. For admitted UDs, the corresponding finite action space is used in the scheduling game (Sec.~\ref{sec:potential_game}).
	}{\color{black}
		Finally, the semantic uncertainty $\rho_{i}^{(t)}\in[0,1]$ combines the calibrated parsing confidence with the compatibility between the inferred service requirements and available execution actions, and is defined as
	}
	\begin{equation}
		\rho_{i}^{(t)}
		=
		\left[
		\alpha_c(1-c_{i}^{\mathsf{conf},(t)})
		+
		\alpha_x r_{i}^{\mathsf{ctx},(t)}
		\right]_0^1,
		\label{eq:semantic_uncertainty}
	\end{equation}
	where {\small$\alpha_c,\alpha_x\ge0$} are weighting coefficients, and {\small$r_{i}^{\mathsf{ctx},(t)}\in[0,1]$} denotes the context-inconsistency risk.
	{\color{black}
		To quantify context inconsistency, STEPS evaluates the mismatch between the original fulfillment requirements and the service outcomes attainable by individually feasible actions. Specifically,
	}
	{\color{black}
		\begin{equation}
			\begin{aligned}
				r_{i}^{\mathsf{ctx},(t)}
				=
				\Bigg[
				\min_{a\in\mathcal{A}_{i}^{\mathsf{ind},(t)}}
				\sum_{\mathsf{k}\in
					\{\mathsf{D},\mathsf{E},\mathsf{M},\mathsf{S}\}}
				\omega_{\mathsf{k}}
				\nu_{i}^{\mathsf{k},(t)}(a)
				\Bigg]_0^1,
			\end{aligned}
			\label{eq:context_risk}
		\end{equation}
	}where {\small$\omega_{\mathsf{D}},\omega_{\mathsf{E}},\omega_{\mathsf{M}},\omega_{\mathsf{S}}\ge0$} are weighting coefficients.
	{\color{black}
		For each candidate action $a$, the mismatch terms are given by
		\begin{equation}
			{\small\begin{aligned}
					\hspace{-1mm}\nu_{i}^{\mathsf{D},(t)}(a)
					&\hspace{-1mm}=\hspace{-1mm}
					\left[
					\frac{D_i^{(t)}(a)-D_{i}^{\mathsf{max},(t)}}
					{D_{i}^{\mathsf{max},(t)}}
					\right]_+,
					\nu_{i}^{\mathsf{E},(t)}(a)
					\hspace{-.5mm}=\hspace{-.5mm}
					\left[
					\frac{E_i^{(t)}(a)\hspace{-.5mm}-\hspace{-.5mm}E_{i}^{\mathsf{max},(t)}}
					{E_{i}^{\mathsf{max},(t)}}
					\right]_+,
					\\
					\hspace{-1mm}\nu_{i}^{\mathsf{M},(t)}(a)
					&\hspace{-.5mm}=\hspace{-1.5mm}
					\left[
					\frac{M_i^{(t)}(a)-M_{i}^{\mathsf{max},(t)}}
					{M_{i}^{\mathsf{max},(t)}}
					\right]_+,
					\nu_{i}^{\mathsf{S},(t)}(a)
					\hspace{-.5mm}=\hspace{-.5mm}
					\left[
					S_{i}^{\mathsf{min},(t)}
					\hspace{-.5mm}-\hspace{-.5mm}S_i^{(t)}(a)
					\right]_+.
			\end{aligned}}
			\label{eq:mismatch_terms}
		\end{equation}
		Here, $[x]_+=\max\{x,0\}$, and
		$D_i^{(t)}(a)$, $E_i^{(t)}(a)$,
		$M_i^{(t)}(a)$, and $S_i^{(t)}(a)$
		denote the service outcomes evaluated for action $a$
		under the current system context.
		The minimization in~\eqref{eq:context_risk}
		uses the same candidate action across all four dimensions,
		thereby measuring the minimum jointly attainable
		contract mismatch.
		A larger $r_i^{\mathsf{ctx},(t)}$ indicates that
		the original semantic requirements are less
		compatible with the available execution actions.
	}Because natural language requests are inherently ambiguous and context-dependent, the semantic contracts are conservatively \textit{calibrated} before scheduling. Let {\small$\beta^{(t)}\ge0$} denote the conservativeness factor at timeslot {\small$t$}, where larger values correspond to more conservative scheduling decisions.
	For each UD {\small$u_i$}, we define \textit{calibrated fulfillment bounds} as
	\begin{equation}
		\hspace{-2mm}
		{\small\begin{aligned}
				\widetilde{D}_{i}^{\mathsf{max},(t)}
				&=
				\frac{D_{i}^{\mathsf{max},(t)}}
				{1+\beta^{(t)}\rho_{i}^{(t)}},~
				\widetilde{E}_{i}^{\mathsf{max},(t)}
				=
				\frac{E_{i}^{\mathsf{max},(t)}}
				{1+\beta^{(t)}\rho_{i}^{(t)}},\\
				\widetilde{M}_{i}^{\mathsf{max},(t)}
				&=
				\frac{M_{i}^{\mathsf{max},(t)}}
				{1+\beta^{(t)}\rho_{i}^{(t)}},~
				\widetilde{S}_{i}^{\mathsf{min},(t)}
				=
				[S_{i}^{\mathsf{min},(t)}
				+\beta^{(t)}\rho_{i}^{(t)}]_0^1,
		\end{aligned}}
		\label{eq:calibrated_bounds}
	\end{equation}
	where {\small$\widetilde{D}_{i}^{\mathsf{max},(t)}$}, {\small$\widetilde{E}_{i}^{\mathsf{max},(t)}$}, {\small$\widetilde{M}_{i}^{\mathsf{max},(t)}$}, and {\small$\widetilde{S}_{i}^{\mathsf{min},(t)}$} denote the scheduling-stage delay, energy, monetary-cost, and T4E requirements, respectively. As the semantic uncertainty {\small$\rho_{i}^{(t)}$} increases, the scheduling process adopts stricter delay, energy, and cost limits while imposing a higher T4E requirement, thereby reserving a larger safety margin against uncertain semantic interpretations. Importantly, the above calibrated bounds are used only during scheduling and do not alter the users' original service requirements (i.e., contract fulfillment is evaluated using the original fulfillment-bound vector {\small$\mathbf{g}_i^{(t)}$}).
	We further introduce a semantic admission threshold {\small$\tau^{\rho,(t)}\in[0,1]$}, which specifies the maximum acceptable semantic uncertainty at timeslot {\small$t$}, and define the set of \textit{schedulable} UDs as
	\begin{equation}
		\mathcal{U}^{\mathsf{s},(t)}
		=
		\left\{
		u_i\in\mathcal{U}^{(t)}
		~\middle|~
		\rho_{i}^{(t)}\le\tau^{\rho,(t)}
		\right\},
		\label{eq:semantic_admission}
	\end{equation}
	where only the UDs in {\small$\mathcal{U}^{\mathsf{s},(t)}$} participate in the subsequent contract-guided scheduling process.

	\subsection{Edge Execution and Contract-Guided Service Loss}
	\label{sec:edge_execution_model}
	\label{sec:service_loss}
	
	For each admitted UD {\small$u_i\in\mathcal{U}^{\mathsf{s},(t)}$}, we define a binary EN selection variable {\small$x_{i,j}^{(t)}\in\{0,1\}$}, where {\small$x_{i,j}^{(t)}=1$} indicates that UD {\small$u_i$}'s task is executed at EN {\small$e_j\in\mathcal{E}_0$}. Since each task can be executed at only one EN, the assignment variables satisfy
	\begin{equation}
		\sum_{e_j\in\mathcal{E}_0}x_{i,j}^{(t)}=1,
		\quad \forall u_i\in\mathcal{U}^{\mathsf{s},(t)}.
		\label{eq:unique_assign}
	\end{equation}
	For each ES {\small$e_j\in\mathcal{E}$}, let {\small$f_{i,j}^{(t)}\ge0$} and {\small$b_{i,j}^{(t)}\ge0$} denote the computing and bandwidth resources allocated to UD {\small$u_i$}, respectively. We note that the aggregate resource consumption at each ES must not exceed its available capacities, yielding
	\begin{equation}
		\hspace{-2mm}\sum_{u_i\in\mathcal{U}^{\mathsf{s},(t)}}\hspace{-1mm}x_{i,j}^{(t)}f_{i,j}^{(t)}\le F_{j}^{(t)},
		\sum_{u_i\in\mathcal{U}^{\mathsf{s},(t)}}\hspace{-.5mm}x_{i,j}^{(t)}b_{i,j}^{(t)}\le B_{j}^{(t)},
		~\forall e_j\in\mathcal{E}.
		\label{eq:edge_capacity}
	\end{equation}
	
	{\color{black}
		Also, to couple EN selection and resource scheduling, we impose the following resource-feasibility constraints
		\begin{equation}
			\begin{aligned}
				0&\le f_{i,j}^{(t)}\le x_{i,j}^{(t)}F_j^{(t)}, \quad
				0\le b_{i,j}^{(t)}\le x_{i,j}^{(t)}B_j^{(t)},\\
				0&\le f_{i,0}^{(t)}\le x_{i,0}^{(t)}F_i^{\mathsf{loc},(t)},
				\quad b_{i,0}^{(t)}=0,
			\end{aligned}
			\label{eq:resource_coupling}
		\end{equation}
		where the first two constraints apply to each $e_j\in\mathcal{E}$. These constraints set resource assignment at unselected ENs to zero and bound allocations at selected ENs by their available capacities. For feasible execution, a selected local action requires positive computing resources, whereas a selected edge action requires positive computing and bandwidth resources. These conditions are satisfied by the resource packages used in scheduling and preserved under proportional scaling when the corresponding capacities are positive.
	}For edge execution, the uplink transmission rate between UD {\small$u_i$} and ES {\small$e_j$} during timeslot {\small$t$} is given by
	{\color{black}
		\begin{equation}
			R_{i,j}^{(t)}
			=
			\begin{cases}
				b_{i,j}^{(t)}
				\log_2\left(
				1+\dfrac{p_i^{(t)}h_{i,j}^{(t)}}{N_0b_{i,j}^{(t)}}
				\right), & x_{i,j}^{(t)}=1,\\
				0, & x_{i,j}^{(t)}=0,
			\end{cases}
			\label{eq:rate}
		\end{equation}
	}where {\small$p_{i}^{(t)}$} denotes the transmit power of UD {\small$u_i$}, {\small$h_{i,j}^{(t)}$} is the channel gain between UD {\small$u_i$} and ES {\small$e_j$}, and {\small$N_0$} is the noise power spectral density. Based on the selected EN and allocated resources, the total service delay is defined as
	{\color{black}
		\begin{equation}
			D_i^{(t)}
			=
			\begin{cases}
				\dfrac{c_i^{(t)}}{f_{i,0}^{(t)}},
				& x_{i,0}^{(t)}=1,\\
				\dfrac{d_i^{(t)}}{R_{i,j}^{(t)}}+
				\dfrac{c_i^{(t)}}{f_{i,j}^{(t)}}+Q_j^{(t)},
				& x_{i,j}^{(t)}=1,\ e_j\in\mathcal{E},
			\end{cases}
			\label{eq:delay}
		\end{equation}
	}{\color{black}
		where $Q_j^{(t)}$ is the pre-decision queueing delay included in the system-side context $\xi^{(t)}$.
	}
	
	For local execution, the energy consumption is modeled as {\small$\kappa_u c_i^{(t)}(f_{i,0}^{(t)})^2$}, where {\small$\kappa_u$} is the effective switched-capacitance of UD {\small$u_i$}'s device~\cite{li2024optimal}. For edge execution, we assume that the dominant energy expenditure arises from uplink transmission as ESs are often not power-limited. Accordingly, the total energy consumption is
	{\color{black}
		\begin{equation}
			E_i^{(t)}
			=
			\begin{cases}
				\kappa_u c_i^{(t)}(f_{i,0}^{(t)})^2,
				& x_{i,0}^{(t)}=1,\\
				\dfrac{p_i^{(t)}d_i^{(t)}}{R_{i,j}^{(t)}},
				& x_{i,j}^{(t)}=1,\ e_j\in\mathcal{E}.
			\end{cases}
			\label{eq:energy}
		\end{equation}
	}{\color{black}
		Let $\lambda_j^{\mathsf{f},(t)}$ and $\lambda_j^{\mathsf{b},(t)}$ denote the unit prices of computing and bandwidth resources at ES $e_j$, respectively, expressed consistently with the corresponding resource allocations. These prices are treated as system-side scheduling parameters at timeslot $t$. The monetary cost incurred by UD $u_i$ is given by
		\begin{equation}
			M_i^{(t)}
			=
			\sum_{e_j\in\mathcal{E}}
			x_{i,j}^{(t)}
			\left(
			\lambda_j^{\mathsf{f},(t)}f_{i,j}^{(t)}
			+
			\lambda_j^{\mathsf{b},(t)}b_{i,j}^{(t)}
			\right).
			\label{eq:monetary_cost}
		\end{equation}
	}
	
	{\color{black}
		Finally, the achieved T4E level of UD $u_i$ is defined as $S_i^{(t)}=x_{i,0}^{(t)}\sigma_{i,0}^{\mathsf{loc},(t)}+\sum_{e_j\in\mathcal{E}}x_{i,j}^{(t)}\sigma_j^{\mathsf{edge},(t)}$, where the local and edge T4E scores are provided as part of the pre-decision system context. In summary, $D_i^{(t)}$, $E_i^{(t)}$, $M_i^{(t)}$, and $S_i^{(t)}$ characterize the realized execution outcomes of a service request and serve as the basis for subsequent contract fulfillment evaluation.
	}
	
	The semantic contract specifies both preferences and fulfillment requirements. To accommodate possible violations of the calibrated fulfillment bounds during scheduling, we introduce the nonnegative slack variables {\small$z_{i}^{\mathsf{D},(t)}$}, {\small$z_{i}^{\mathsf{E},(t)}$}, {\small$z_{i}^{\mathsf{M},(t)}$}, and {\small$z_{i}^{\mathsf{S},(t)}$} corresponding to delay, energy consumption, monetary cost, and T4E, respectively. We impose the following conditions
	{\color{black}
		\begin{equation}
			{\small\begin{aligned}
					D_i^{(t)}&\le\widetilde{D}_i^{\mathsf{max},(t)}
					(1+z_i^{\mathsf{D},(t)}),~~
					E_i^{(t)}\le\widetilde{E}_i^{\mathsf{max},(t)}
					(1+z_i^{\mathsf{E},(t)}),\\
					M_i^{(t)}&\le\widetilde{M}_i^{\mathsf{max},(t)}
					(1+z_i^{\mathsf{M},(t)}),~~
					\widetilde{S}_i^{\mathsf{min},(t)}-S_i^{(t)}
					\le z_i^{\mathsf{S},(t)}.
			\end{aligned}}
			\label{eq:slack_constraints}
		\end{equation}
	}
	
	In essence, the slack variables quantify the extent to which the realized service outcomes deviate from the calibrated contract requirements. Based on these quantities, we define the \textit{semantic contract violation degree} as {\small$\Omega_{i}^{\mathsf{sc},(t)}=\zeta_{\mathsf{D}}z_{i}^{\mathsf{D},(t)}+\zeta_{\mathsf{E}}z_{i}^{\mathsf{E},(t)}+\zeta_{\mathsf{M}}z_{i}^{\mathsf{M},(t)}+\zeta_{\mathsf{S}}z_{i}^{\mathsf{S},(t)}$}, where {\small$\zeta_{\mathsf{D}}$, $\zeta_{\mathsf{E}}$}, {\small$\zeta_{\mathsf{M}}$}, and {\small$\zeta_{\mathsf{S}}$} are nonnegative violation-penalty coefficients.
	Based on the calibrated semantic contract, we define \textit{contract-guided service loss} of UD {\small$u_i$} at timeslot {\small$t$} as
	\begin{equation}
		\begin{aligned}
			J_{i}^{(t)}
			=&~
			w_{i}^{\mathsf{D},(t)}
			\frac{D_{i}^{(t)}}{\widetilde{D}_{i}^{\mathsf{max},(t)}}
			+
			w_{i}^{\mathsf{E},(t)}
			\frac{E_{i}^{(t)}}{\widetilde{E}_{i}^{\mathsf{max},(t)}}
			\\
			&+
			w_{i}^{\mathsf{M},(t)}
			\frac{M_{i}^{(t)}}{\widetilde{M}_{i}^{\mathsf{max},(t)}}
			+
			w_{i}^{\mathsf{S},(t)}
			(1-S_{i}^{(t)})
			+
			\eta\Omega_{i}^{\mathsf{sc},(t)},
		\end{aligned}
		\label{eq:J_loss}
	\end{equation}
	where {\small$\eta\ge0$} controls the importance of contract violations. In \eqref{eq:J_loss}, the first three terms capture the normalized delay, energy, and monetary-cost performance relative to the calibrated fulfillment bounds, {\color{black}while the fourth term penalizes lower achieved T4E.} Also, the final term explicitly penalizes violations of the calibrated semantic contract.
	
	\subsection{Problem Formulation}
	\label{sec:problem_formulation}
	
	Let {\small$\mathbf{z}_{i}^{(t)}=[z_{i}^{\mathsf{D},(t)},z_{i}^{\mathsf{E},(t)},z_{i}^{\mathsf{M},(t)},z_{i}^{\mathsf{S},(t)}]^{\top}$} denote the semantic contract violation vector of UD {\small$u_i$} at timeslot {\small$t$}. Based on the semantic contract framework developed above, we formulate the long-term contract-guided scheduling as problem {\small$\mathcal{P}$}:
	\begin{equation}
		\small
		\begin{aligned}
			(\mathcal{P}):~\quad
			\min_{\{\mathbf{x}^{(t)},\mathbf{f}^{(t)},\mathbf{b}^{(t)},\mathbf{z}^{(t)}\}_{t\in\mathcal{T}}}
			~
			\frac{1}{T}
			\sum_{t=1}^{T}
			\sum_{u_i\in\mathcal{U}^{\mathsf{s},(t)}}
			J_{i}^{(t)}
		\end{aligned}
		\label{eq:P0}
		\vspace{-3mm}
	\end{equation}
	\noindent\textbf{s.t.}
	
	\noindent\textbf{\textit{\underline{Constraints}:}}
	\begin{itemize}
		\item \textit{Task assignment:} \eqref{eq:unique_assign}
		\item \textit{ES resource capacity:} \eqref{eq:edge_capacity}
		\item \textit{EN selection and resource-allocation coupling:} \eqref{eq:resource_coupling}
		\item \textit{Calibrated semantic-contract fulfillment:} \eqref{eq:slack_constraints}
	\end{itemize}
	
	\noindent\textbf{\textit{\underline{Variables}:}}
	\begin{itemize}
		\item \textit{EN selection variables:}\\
		$\{x_{i,j}^{(t)}\in\{0,1\}\mid u_i\in\mathcal{U}^{\mathsf{s},(t)}, e_j\in\mathcal{E}_0, t\in\mathcal{T}\}$
		\item \textit{Computing and bandwidth allocation variables:}\\
		$\{\mathbf{f}^{(t)},\mathbf{b}^{(t)}\}_{t\in\mathcal{T}}$ satisfying the resource-domain and coupling constraints in~\eqref{eq:resource_coupling}
		\item \textit{Semantic contract violation variables:}\\
		$\{\mathbf{z}_{i}^{(t)}\succeq\mathbf{0}\mid u_i\in\mathcal{U}^{\mathsf{s},(t)}, t\in\mathcal{T}\}$.\footnote{$\mathbf{a}\succeq\mathbf{b}$ denotes component-wise inequality for vectors of the same dimension; hence, $\mathbf{z}_{i}^{(t)}\succeq\mathbf{0}$ means that all entries of $\mathbf{z}_{i}^{(t)}$ are nonnegative.}
	\end{itemize}
	
	\noindent In {\small$\mathcal{P}$}, the objective minimizes the long-term average contract-guided service loss across all admitted UDs. The listed constraints enforce task assignment, ES-side capacity feasibility, EN selection/resource-allocation coupling, and calibrated semantic-contract fulfillment. The listed variables correspond to EN selection, communication-computation resource allocation, and semantic contract violation degrees. Problem {\small$\mathcal{P}$} is a long-term mixed-integer nonlinear program (MINLP) that jointly addresses task assignment, communication-resource allocation, and computation-resource provisioning. The combinatorial nature of EN selection, together with the nonlinear transmission-rate and delay expressions, renders its direct solution computationally prohibitive in dynamic E4NetAI environments. This motivates the development of STEPS, transforming the original long-term optimization into a distributed and adaptive scheduling framework as discussed next.

	\section{Design Methodology of STEPS}
	\label{sec:steps}
	
	\noindent This section introduces STEPS, an online distributed framework for semantic contract-guided edge scheduling.

	\subsection{Per-Slot Penalized Surrogate and Potential Game}
	\label{sec:potential_game}
	
	Directly solving {\small$\mathcal{P}$} over the entire time horizon requires knowledge of future UDs' requests, wireless channels, edge workloads, and resource availability, as well as centralized coordination across all UDs and ESs. Such knowledge is often unavailable in dynamic E4NetAI environments, where service demands and resource states evolve over time~\cite{li2024optimal,zhang2025beyond}. Accordingly, STEPS adopts a slot-wise online decision structure based only on the currently observed system state. Under this online formulation, the service loss experienced by a UD is primarily determined by its own EN selection and resource allocation decision. However, because multiple UDs may compete for the limited resources of the same ESs, their decisions remain coupled through shared computing and bandwidth capacities. To enable distributed scheduling, STEPS replaces these global capacity couplings with congestion-dependent penalty terms that reflect the impact of aggregate resource usage at each ES. Consequently, the original coupled scheduling problem is transformed into a penalized surrogate formulation in which the effects of resource contention can be locally observed and incorporated into individual decision making. This construction follows the general principles of congestion games and resource-pricing mechanisms~\cite{rosenthal1973class,kelly1998rate}.
	
	To further facilitate distributed optimization, the available computing and bandwidth resources are discretized into finite \textit{resource packages}~\cite{zhang2025beyond}. At each timeslot, every admitted UD {\small$u_i\in\mathcal{U}^{\mathsf{s},(t)}$} selects an intended EN with a corresponding resource-package combination according to its semantic contract and the current system state. The ESP then performs feasibility verification and allocation realization based on the actual resource capacities and possible resource conflicts. 
	{\color{black}To formalize this process, we first define the individually feasible action set of each UD and then formulate the per-slot penalized surrogate as an exact potential game.}
	
	\begin{Defn}(Action Space of UDs)
		{\color{black}For each incoming UD $u_i\in\mathcal{U}^{(t)}$, its candidate action is denoted by $a_i^{(t)}=(e_i^{(t)},r_i^{\mathsf{f},(t)},r_i^{\mathsf{b},(t)})$,} where {\small$e_i^{(t)}\in\mathcal{E}_0$}, {\small$r_i^{\mathsf{f},(t)}$}, and {\small$r_i^{\mathsf{b},(t)}$} capture the selected EN, computing-resource package, and bandwidth-resource package, respectively. 
		Since local execution does not require uplink bandwidth, {\color{black}the individually feasible action set and the corresponding scheduling action space are defined as}
		{\color{black}
			\begin{equation}
				{\small\begin{aligned}
						\hspace{-2mm}\mathcal{A}_i^{\mathsf{ind},(t)}
						=\mathcal{A}_i^{\mathsf{loc},(t)}
						\cup\mathcal{A}_i^{\mathsf{edge},(t)},~~
						\mathcal{A}_i^{(t)}
						=\mathcal{A}_i^{\mathsf{ind},(t)},
						~ u_i\in\mathcal{U}^{\mathsf{s},(t)}.
				\end{aligned}}
				\label{eq:action_space_union}
			\end{equation}
		}
		where {\small$\mathcal{A}_{i}^{\mathsf{loc},(t)}=\{(e_0,r^{\mathsf{f}},0)\mid r^{\mathsf{f}}\in\mathcal{R}_{i}^{\mathsf{loc},\mathsf{f},(t)}\}$} and {\small$\mathcal{A}_{i}^{\mathsf{edge},(t)}=\{(e_j,r^{\mathsf{f}},r^{\mathsf{b}})\mid e_j\in\mathcal{E},~r^{\mathsf{f}}\in\mathcal{R}_{i,j}^{\mathsf{f},(t)},~r^{\mathsf{b}}\in\mathcal{R}_{i,j}^{\mathsf{b},(t)}\}$}. Here, {\small$\mathcal{R}_{i}^{\mathsf{loc},\mathsf{f},(t)}$}, {\small$\mathcal{R}_{i,j}^{\mathsf{f},(t)}$} and {\small$\mathcal{R}_{i,j}^{\mathsf{b},(t)}$} denote finite sets of feasible computing and bandwidth resource packages. Consequently, each action jointly specifies an execution location together with an associated resource allocation configuration, and can be mapped directly to the corresponding EN selection variable {\small$x_{i,j}^{(t)}$} and the selected resource-package values.
	\end{Defn}
	
	For notational clarity, for any action $a_i^{(t)}=(e_i^{(t)},r_i^{\mathsf f,(t)},r_i^{\mathsf b,(t)})$, let $e(a_i^{(t)})$ denote its selected execution node and $r_i^{\mathsf z,(t)}(a_i^{(t)})$ denote its selected resource package of type $\mathsf z\in\{\mathsf f,\mathsf b\}$. If $e(a_i^{(t)})=e_j\in\mathcal E$, we define $j(a_i^{(t)})=j$ as the index of the selected ES. Given an action profile {\small$\mathbf{a}^{(t)}=[a_i^{(t)}]_{u_i\in\mathcal{U}^{\mathsf{s},(t)}}$}, the aggregate computing and bandwidth loads induced at ES {\small$e_j$} are
	\begin{equation}
		{\small\begin{aligned}
				L_j^{\mathsf{z},(t)}(\mathbf{a}^{(t)})
				&=
				\sum_{u_i\in\mathcal{U}^{\mathsf{s},(t)}}
				\mathbb{I}_{\{e(a_i^{(t)})=e_j\}}
				r_i^{\mathsf{z},(t)}(a_i^{(t)}),
				\quad
				\mathsf z\in\{\mathsf f,\mathsf b\}.
		\end{aligned}}
		\label{eq:intended_compute_load}
	\end{equation}
	where {\small$\mathbb{I}_{\{\cdot\}}$} is the indicator function.

	{\color{black}For each candidate action, STEPS evaluates the corresponding service metrics relative to the calibrated semantic contract.} {\color{black}They are computed using the selected resource package and fixed pre-decision context, before resource realization.} Thus, we define the action-dependent fulfillment-shortfall terms as
	\begin{equation}
		\hspace{-3mm}	{\small\begin{aligned}
				&z_i^{\mathsf{D},(t)}(a_i)
				\hspace{-0.5mm}=\hspace{-0.5mm}
				\left[
				\frac{D_i^{(t)}(a_i)}
				{\widetilde{D}_i^{\mathsf{max},(t)}}
				-1
				\right]_+,
				z_i^{\mathsf{E},(t)}(a_i)
				\hspace{-0.5mm}=\hspace{-0.5mm}
				\left[
				\frac{E_i^{(t)}(a_i)}
				{\widetilde{E}_i^{\mathsf{max},(t)}}
				-1
				\right]_+,
				\\&
				z_i^{\mathsf{M},(t)}(a_i)
				\hspace{-0.5mm}=\hspace{-0.5mm}
				\left[
				\frac{M_i^{(t)}(a_i)}
				{\widetilde{M}_i^{\mathsf{max},(t)}}
				-1
				\right]_+,
				z_i^{\mathsf{S},(t)}(a_i)
				\hspace{-0.5mm}=\hspace{-0.5mm}
				\left[
				\widetilde{S}_i^{\mathsf{min},(t)}
				\hspace{-0.5mm}-\hspace{-0.5mm}
				S_i^{(t)}(a_i)
				\right]_+.
		\end{aligned}}
		\label{eq:action_slack}
	\end{equation}
	These quantities measure the extent to which a candidate action violates the calibrated contract requirements. Accordingly, the \textit{action-dependent semantic contract violation degree} is defined as {\small$\Omega_i^{\mathsf{sc},(t)}(a_i)=\zeta_{\mathsf{D}}z_i^{\mathsf{D},(t)}(a_i)+\zeta_{\mathsf{E}}z_i^{\mathsf{E},(t)}(a_i)+\zeta_{\mathsf{M}}z_i^{\mathsf{M},(t)}(a_i)+\zeta_{\mathsf{S}}z_i^{\mathsf{S},(t)}(a_i)$}. Unlike the optimization in Sec.~\ref{sec:problem_formulation}, this formulation evaluates contract violations directly from candidate actions, and thus avoids introducing additional continuous slack variables into the finite game. Subsequently, we define the \textit{contract-guided service loss} of UD {\small$u_i$} as
	\begin{equation}
		{\small\begin{aligned}
				&H_i^{(t)}(a_i^{(t)})
				=~
				w_i^{\mathsf{D},(t)}
				\frac{D_i^{(t)}(a_i^{(t)})}
				{\widetilde{D}_i^{\mathsf{max},(t)}}
				+
				w_i^{\mathsf{E},(t)}
				\frac{E_i^{(t)}(a_i^{(t)})}
				{\widetilde{E}_i^{\mathsf{max},(t)}}
				\\
				&+
				w_i^{\mathsf{M},(t)}
				\frac{M_i^{(t)}(a_i^{(t)})}
				{\widetilde{M}_i^{\mathsf{max},(t)}}
				\hspace{-.5mm}+\hspace{-.5mm}
				w_i^{\mathsf{S},(t)}
				\big(1-S_i^{(t)}(a_i^{(t)})\big)
				\hspace{-.5mm}+\hspace{-.5mm}
				\eta
				\Omega_i^{\mathsf{sc},(t)}(a_i^{(t)}),
		\end{aligned}}
		\label{eq:H_loss}
	\end{equation}
	where the first four terms evaluate delay, energy consumption, monetary cost, and T4E performance relative to the calibrated semantic contract, while the final term penalizes explicit contract violations. Therefore, {\small$H_i^{(t)}(\cdot)$} given by \eqref{eq:H_loss} captures the semantic objective of each individual UD.
	
	We then note that while {\small$H_i^{(t)}(\cdot)$} reflects user-side objectives, it does not account for resource contention among UDs sharing the same ES. Subsequently, to discourage excessive concentration of requests on ESs, we introduce congestion penalties based on aggregate resource utilization. {\color{black}The ESP evaluates these penalties using the aggregate ES loads during coordinated best-response updates.} Specifically, for {\small$e_j$}, the computing and bandwidth congestion (collectively capturing the resource congestion) penalties are defined as
	\begin{equation}
		\hspace{-0mm}
		{\small\begin{aligned}
				\Gamma_j^{\mathsf{f},(t)}(L)
				\hspace{-.5mm}=\hspace{-.5mm}
				\frac{\gamma_{\mathsf{f}}}{2}
				\left(
				\left[
				\frac{L}{F_j^{(t)}}
				-
				\varrho_{\mathsf{f}}
				\right]_+
				\right)^2,~
				\Gamma_j^{\mathsf{b},(t)}(L)
				\hspace{-.5mm}=\hspace{-.5mm}
				\frac{\gamma_{\mathsf{b}}}{2}
				\left(
				\left[
				\frac{L}{B_j^{(t)}}
				-
				\varrho_{\mathsf{b}}
				\right]_+
				\right)^2,
		\end{aligned}}
		\label{eq:congestion_penalty}
	\end{equation}
	where {\small$\gamma_{\mathsf{f}}$} and {\small$\gamma_{\mathsf{b}}$} are nonnegative penalty coefficients, and {\small$\varrho_{\mathsf{f}},\varrho_{\mathsf{b}}\in(0,1]$} are target utilization thresholds. In \eqref{eq:congestion_penalty}, when {\small$\varrho_{\mathsf{f}}=\varrho_{\mathsf{b}}=1$}, penalties are activated only after capacity limits are exceeded. Also, when {\small$\varrho_{\mathsf{f}}$} and {\small$\varrho_{\mathsf{b}}$} are smaller than one, the system proactively discourages high-utilization ESs before overloading occurs. With both the user-side service loss and ES-side congestion penalties defined above, we next formulate the per-timeslot scheduling interaction as a distributed game.
	
	\begin{Defn}(Contract-Guided Scheduling Game)
		At each timeslot {\small$t$}, the distributed scheduling interaction among admitted UDs is modeled as the contract-guided scheduling game
		\begin{equation}
			\mathcal{G}^{(t)}
			=
			\left(
			\mathcal{U}^{\mathsf{s},(t)},
			\{\mathcal{A}_{i}^{(t)}\}_{u_i\in\mathcal{U}^{\mathsf{s},(t)}},
			\{\mathcal{J}_{i}^{(t)}\}_{u_i\in\mathcal{U}^{\mathsf{s},(t)}}
			\right),
		\end{equation}
		where {\small$\mathcal{U}^{\mathsf{s},(t)}$} denotes the set of admitted UDs and constitutes the player set, {\small$\mathcal{A}_{i}^{(t)}$} is the finite action space of UD {\small$u_i$}, and {\small$\mathcal{J}_{i}^{(t)}$} is the scheduling cost incurred by UD {\small$u_i$}.
	\end{Defn}
	
	Given the game defined above, we combine the contract-guided service loss and the ES-side marginal congestion impact into the scheduling cost of each player. Specifically, the scheduling cost of UD {\small$u_i$} is defined as
	{\color{black}
		\begin{equation}{\small
				\begin{aligned}
					\hspace{-2.5mm}\mathcal{J}_{i}^{(t)}(a_i^{(t)},\mathbf{a}_{-i}^{(t)})\hspace{-.5mm}
					=\hspace{-.5mm}
					H_i^{(t)}(a_i^{(t)})
					\hspace{-.5mm}+\hspace{-.5mm}
					(1\hspace{-.5mm}+\hspace{-.5mm}\chi^{(t)})
					\hspace{-1mm}\sum_{\mathsf z\in\{\mathsf f,\mathsf b\}}\hspace{-1mm}
					\Delta\Gamma_i^{\mathsf z,(t)}
					(a_i^{(t)},\mathbf{a}_{-i}^{(t)}),
			\end{aligned}}
			\label{eq:user_cost}
		\end{equation}
	}
	where {\small$\mathbf{a}_{-i}^{(t)}$} denotes the actions of all admitted UDs except {\small$u_i$}, and {\small$\chi^{(t)}$} is the edge coordination gain. For each ES {\small$e_j$} and each resource type {\small$\mathsf z\in\{\mathsf f,\mathsf b\}$}, the aggregate load excluding UD {\small$u_i$} is defined as
	\begin{equation}
		{\small\begin{aligned}
				L_{j,-i}^{\mathsf z,(t)}(\mathbf a_{-i}^{(t)})
				=
				\sum_{u_k\in\mathcal U^{\mathsf s,(t)}\setminus\{u_i\}}
				\mathbb{I}_{\{e(a_k^{(t)})=e_j\}}
				r_k^{\mathsf z,(t)}(a_k^{(t)}).
		\end{aligned}}
		\label{eq:load_excluding_i}
	\end{equation}
	The marginal congestion cost induced by UD {\small$u_i$} under candidate action {\small$a_i^{(t)}$} is defined as the incremental ES-side congestion penalty caused by adding {\small$u_i$}'s selected resource package to the ES selected by this action. Specifically, if {\small$e(a_i^{(t)})=e_j\in\mathcal E$}, we define {\small$\Delta\Gamma_i^{\mathsf z,(t)}(a_i^{(t)},\mathbf a_{-i}^{(t)})\triangleq\Gamma_j^{\mathsf z,(t)}\!\left(L_{j,-i}^{\mathsf z,(t)}(\mathbf a_{-i}^{(t)})+r_i^{\mathsf z,(t)}(a_i^{(t)})\right)-\Gamma_j^{\mathsf z,(t)}\!\left(L_{j,-i}^{\mathsf z,(t)}(\mathbf a_{-i}^{(t)})\right),~\mathsf z\in\{\mathsf f,\mathsf b\}$}. For local execution (i.e., {\small$e(a_i^{(t)})=e_0$}), we set {\small$\Delta\Gamma_i^{\mathsf z,(t)}(a_i^{(t)},\mathbf a_{-i}^{(t)})=0$}. Here, {\small$\Gamma_{j}^{\mathsf z,(t)}(\cdot)$} is the ES-specific congestion penalty function defined in~\eqref{eq:congestion_penalty}. With the above marginal congestion construction, the \textit{potential function} associated with {\small$\mathcal{G}^{(t)}$} is defined as
	\begin{equation}
		{\small\begin{aligned}
				&\hspace{-1mm}\Phi^{(t)}(\mathbf{a}^{(t)})
				=
				\sum_{u_i\in\mathcal{U}^{\mathsf{s},(t)}}
				H_i^{(t)}(a_i^{(t)})
				\\
				&\hspace{-1mm}+
				(1+\chi^{(t)})\hspace{-1mm}
				\sum_{e_j\in\mathcal{E}}
				\left[
				\Gamma_j^{\mathsf{f},(t)}
				\left(
				L_j^{\mathsf{f},(t)}(\mathbf{a}^{(t)})
				\right)
				+
				\Gamma_j^{\mathsf{b},(t)}
				\left(
				L_j^{\mathsf{b},(t)}(\mathbf{a}^{(t)})
				\right)
				\right].
		\end{aligned}}
		\label{eq:potential}
	\end{equation}
	The first term in {\small$\Phi^{(t)}$} aggregates the contract-guided service losses of admitted UDs, while the second term penalizes edge-side resource congestion. Accordingly, the potential game targets the following finite-action penalized surrogate problem:
	\begin{equation}
		\mathcal{P}_{\mathsf{pg}}^{(t)}:\quad
		\min_{\mathbf{a}^{(t)}\in\mathcal{A}^{(t)}}
		\Phi^{(t)}(\mathbf{a}^{(t)}),
		\quad
		\mathcal{A}^{(t)}
		\triangleq
		\prod_{u_i\in\mathcal{U}^{\mathsf{s},(t)}}
		\mathcal{A}_i^{(t)}.
		\label{eq:pg_surrogate_problem}
	\end{equation}
	
	Compared with the original long-term problem {\small$\mathcal{P}$}, {\small$\mathcal{P}_{\mathsf{pg}}^{(t)}$} focuses on the current timeslot, restricts each admitted UD to the feasible action space in~\eqref{eq:action_space_union}, and replaces hard ES-side capacity constraints with the congestion penalties in~\eqref{eq:congestion_penalty}. Thus, the potential game does not directly solve the long-horizon MINLP {\small$\mathcal{P}$}; instead, it provides a distributed local-improvement mechanism for the per-slot surrogate objective in~\eqref{eq:pg_surrogate_problem}.
	{\color{black}To solve the per-slot surrogate, STEPS adopts an ESP-coordinated \textit{asynchronous strict best-response process}.} Let {\small$\mathbf{a}^{(t)}(r)$} denote the current action profile at update round {\small$r$}. {\color{black}During each round, admitted UDs are visited sequentially. For each visited UD $u_i$, the best response is computed using the current ES-load information maintained by the ESP as}
	\begin{equation}
		a_i^{\mathsf{br},(t)}
		\in
		\arg\min_{a_i^{(t)}\in\mathcal{A}_i^{(t)}}
		\mathcal{J}_i^{(t)}(a_i^{(t)},\mathbf{a}_{-i}^{(t)}).
		\label{eq:best_response_update}
	\end{equation}
	
	{\color{black}
		An action update is accepted only if the best response reduces the current scheduling cost by more than $\varepsilon_{\mathsf{br}}$; otherwise, the current action is retained. Following each accepted update, the ES loads in~\eqref{eq:intended_compute_load} and congestion penalties in~\eqref{eq:congestion_penalty} are refreshed before the next UD is visited. The process terminates when a complete round produces no accepted improvement or the maximum round number $I_{\mathsf{max}}$ is reached.
	}
	
	We note that the resulting action profile specifies the intended EN selections and resource-allocation requests of all admitted UDs before physical resource realization. Since these intended requests are determined through distributed best-response updates, their aggregate demand may still exceed the available resources of some ESs. Thus, to guarantee physical feasibility before execution, the ESP performs a resource-realization step on the final game outcome. Let {\small$\mathbf{a}^{\mathsf{fin},(t)}$} denote the final intended UD action profile returned by the asynchronous best-response procedure. For each ES {\small$e_j$}, we define the computing- and bandwidth-resource realization factors as
	\begin{equation}
		\kappa_j^{\mathsf{f},(t)}
		=
		\begin{cases}
			\min\left\{1,\frac{F_j^{(t)}}
			{L_j^{\mathsf{f},(t)}(\mathbf{a}^{\mathsf{fin},(t)})}\right\},
			&
			L_j^{\mathsf{f},(t)}(\mathbf{a}^{\mathsf{fin},(t)})>0,\\
			1,
			&
			L_j^{\mathsf{f},(t)}(\mathbf{a}^{\mathsf{fin},(t)})=0,
		\end{cases}
		\label{eq:compute_realization_factor}
	\end{equation}
	and
	\begin{equation}
		\kappa_j^{\mathsf{b},(t)}
		=
		\begin{cases}
			\min\left\{1,\frac{B_j^{(t)}}
			{L_j^{\mathsf{b},(t)}(\mathbf{a}^{\mathsf{fin},(t)})}\right\},
			&
			L_j^{\mathsf{b},(t)}(\mathbf{a}^{\mathsf{fin},(t)})>0,\\
			1,
			&
			L_j^{\mathsf{b},(t)}(\mathbf{a}^{\mathsf{fin},(t)})=0.
		\end{cases}
		\label{eq:bandwidth_realization_factor}
	\end{equation}
	These factors proportionally scale the requested resources whenever the aggregate demand exceeds the available capacity. Accordingly, the realized resource allocation is given by
	\begin{equation}
		{\small\begin{aligned}
				f_{i,j}^{(t)}
				&=
				\mathbb{I}_{\{e_i^{\mathsf{fin},(t)}=e_j\}}
				\kappa_j^{\mathsf{f},(t)}
				r_i^{\mathsf{f},\mathsf{fin},(t)},\\
				b_{i,j}^{(t)}
				&=
				\mathbb{I}_{\{e_i^{\mathsf{fin},(t)}=e_j\}}
				\kappa_j^{\mathsf{b},(t)}
				r_i^{\mathsf{b},\mathsf{fin},(t)},
				\quad e_j\in\mathcal{E}.
		\end{aligned}}
		\label{eq:realized_resource_allocation}
	\end{equation}
	Similarly, for local execution, the realized local computing resource is {\small$f_{i,0}^{(t)}=\mathbb{I}_{\{e_i^{\mathsf{fin},(t)}=e_0\}}r_i^{\mathsf{f},\mathsf{fin},(t)}$}, while {\small$b_{i,0}^{(t)}=0$}. By construction, the realized allocation satisfies
	\begin{equation}
		\sum_{u_i\in\mathcal{U}^{\mathsf{s},(t)}}
		f_{i,j}^{(t)}
		\le F_j^{(t)},
		\quad
		\sum_{u_i\in\mathcal{U}^{\mathsf{s},(t)}}
		b_{i,j}^{(t)}
		\le B_j^{(t)},
		\quad \forall e_j\in\mathcal{E}.
		\label{eq:realized_capacity_feasibility}
	\end{equation}
	The final allocation in~\eqref{eq:realized_resource_allocation} is then used to compute the actual execution delay, energy consumption, monetary cost, and T4E fulfillment. In this manner, the potential game determines distributed scheduling intentions, whereas the realization step converts them into a resource-feasible execution plan. 
	
	\subsection{Fulfillment-Driven Feedback and Adaptive Optimization}
	\label{sec:evolutive_feedback}
	
	After task execution, the ESP collects the realized service outcomes from the selected ESs for edge execution and from the originating UDs for local execution, and then evaluates how well the original semantic contract is satisfied. Unlike {\small$\Omega_{i}^{\mathsf{sc},(t)}$}, which is used during scheduling with calibrated fulfillment bounds, the post-execution evaluation is performed with respect to the original fulfillment bounds in {\small$\mathbf{g}_{i}^{(t)}$}. Specifically, the fulfillment deviation of admitted UD {\small$u_i$} is defined as
	\begin{equation}
		{\small\begin{aligned}
				\Omega_{i}^{\mathsf{ful},(t)}
				&=~
				\zeta_{\mathsf{D}}
				\left[
				\frac{D_{i}^{(t)}-D_{i}^{\mathsf{max},(t)}}
				{D_{i}^{\mathsf{max},(t)}}
				\right]_+
				+
				\zeta_{\mathsf{E}}
				\left[
				\frac{E_{i}^{(t)}-E_{i}^{\mathsf{max},(t)}}
				{E_{i}^{\mathsf{max},(t)}}
				\right]_+
				\\
				&+
				\zeta_{\mathsf{M}}
				\left[
				\frac{M_{i}^{(t)}-M_{i}^{\mathsf{max},(t)}}
				{M_{i}^{\mathsf{max},(t)}}
				\right]_+
				+
				\zeta_{\mathsf{S}}
				\left[
				S_{i}^{\mathsf{min},(t)}
				-
				S_{i}^{(t)}
				\right]_+.
		\end{aligned}}
		\label{eq:ful_deviation}
	\end{equation}
	The semantic contract fulfillment degree is then defined as
	{\small$\varphi_{i}^{(t)}=\exp(-\Omega_{i}^{\mathsf{ful},(t)})$}.
	A larger {\small$\varphi_{i}^{(t)}$} indicates better fulfillment; with positive violation-penalty coefficients, {\small$\varphi_{i}^{(t)}=1$} indicates full contract satisfaction.
	
	The above-defined fulfillment degree provides a direct measure of how well the delivered service aligns with the user's original semantic expectations. Since both user requirements and system operating conditions may evolve over time, fulfillment outcomes can exhibit changes that should be reflected in future scheduling decisions. To enable such adaptation in non-stationary E4NetAI environments, \textit{STEPS constructs feedback signals from observed fulfillment outcomes and distinguishes two complementary sources of variation.} The first captures changes in user-side semantic requirements over time, whereas the second captures variations in the system's ability to satisfy semantic contracts under evolving execution conditions. The following definitions formalize these two sources of variation.

	\begin{Defn}(Semantic Request Drift)
		% Semantic request drift quantifies temporal changes in user-side service requirements and 
		% as reflected by the generated semantic contracts. Specifically, it captures variations in service preferences, fulfillment-bound stringency, and semantic uncertainty. To ensure that changes in user semantics are not masked by the semantic-admission mechanism, the drift 
		% is computed over all incoming UDs rather than only the admitted ones. 
		% To characterize the semantic requirements of UD {\small$u_i$},
		Let {\small $\mathbf{s}_{i}^{(t)}$} denote 
		the semantic state of each UD {\small$u_i$} defined as
		\begin{equation}
			\mathbf{s}_{i}^{(t)}
			=
			\left[
			(\mathbf{w}_{i}^{(t)})^{\top},
			(\bar{\mathbf{g}}_{i}^{(t)})^{\top},
			\rho_{i}^{(t)}
			\right]^{\top},
			\label{eq:semantic_state}
		\end{equation}
		where {\small$\bar{\mathbf{g}}_{i}^{(t)}$} is the normalized version of the fulfillment-bound vector {\small$\mathbf{g}_{i}^{(t)}$}. The aggregate semantic state of the incoming requests across all UDs at timeslot {\small$t$} is then given by
		\begin{equation}
			\bar{\mathbf{s}}_{\mathsf{all}}^{(t)}
			=
			\begin{cases}
				\frac{1}{|\mathcal{U}^{(t)}|}
				\sum_{u_i\in\mathcal{U}^{(t)}}
				\mathbf{s}_{i}^{(t)},
				&
				|\mathcal{U}^{(t)}|>0,
				\\
				\bar{\mathbf{s}}_{\mathsf{all}}^{{(t-1)}},
				&
				|\mathcal{U}^{(t)}|=0.
			\end{cases}
			\label{eq:average_semantic_state}
		\end{equation}
		To detect temporal changes in user semantics, STEPS compares the average semantic states observed over a recent window and a historical window. In particular, let {\small$W^{\mathsf{new}}$} and {\small$W^{\mathsf{old}}$} denote the recent and historical time windows, respectively\footnote{For early timeslots, when a complete historical window is unavailable, the averages are computed using the available observations.}. The semantic request drift is defined as
		\begin{equation}
			\Delta_{\mathsf{sem}}^{(t)}
			=
			\left\|
			\frac{1}{|W^{\mathsf{new}}|}
			\sum_{\tau\in W^{\mathsf{new}}}
			\bar{\mathbf{s}}_{\mathsf{all}}^{(\tau)}
			-
			\frac{1}{|W^{\mathsf{old}}|}
			\sum_{\tau\in W^{\mathsf{old}}}
			\bar{\mathbf{s}}_{\mathsf{all}}^{(\tau)}
			\right\|_2,
			\label{eq:semantic_drift}
		\end{equation}
		where a larger value of {\small$\Delta_{\mathrm{sem}}^{(t)}$} indicates a more pronounced shift in the semantic characteristics of incoming service requests.
	\end{Defn}
	
	\begin{Defn}(Contract-Fulfillment Drift)
		% Unlike semantic request drift, which reflects changes in user-side requirements, fulfillment drift is derived from post-execution contract deviations and captures variations in runtime service capability.
		Let {\small $P_{\mathsf{ful}}^{(t)}$} denote the average fulfillment pressure among admitted UDs defined as
		\begin{equation}
			P_{\mathsf{ful}}^{(t)}
			=
			\begin{cases}
				\frac{1}{|\mathcal{U}^{\mathsf{s},(t)}|}
				\sum_{u_i\in\mathcal{U}^{\mathsf{s},(t)}}
				\Omega_{i}^{\mathsf{ful},(t)},
				&
				|\mathcal{U}^{\mathsf{s},(t)}|>0,
				\\
				0,
				&
				|\mathcal{U}^{\mathsf{s},(t)}|=0,
			\end{cases}
			\label{eq:fulfillment_pressure}
		\end{equation}
		% To characterize it, we first define the average fulfillment pressure among admitted UDs as
		where a larger {\small$P_{\mathrm{ful}}^{(t)}$} indicates that the system experiences greater difficulty in satisfying the semantic contracts of admitted UDs. To identify the changes in contract-fulfillment capability, STEPS compares the average fulfillment pressure observed over a recent window and a historical window, where the (contract-)fulfillment drift is defined as
		\begin{equation}
			\Delta_{\mathsf{ful}}^{(t)}
			=
			\left|
			\frac{1}{|W^{\mathsf{new}}|}
			\sum_{\tau\in W^{\mathsf{new}}}
			P_{\mathsf{ful}}^{(\tau)}
			-
			\frac{1}{|W^{\mathsf{old}}|}
			\sum_{\tau\in W^{\mathsf{old}}}
			P_{\mathsf{ful}}^{(\tau)}
			\right|.
			\label{eq:fulfillment_drift}
		\end{equation}
		A larger value of {\small$\Delta_{\mathrm{ful}}^{(t)}$} indicates a more pronounced change in the system's contract-fulfillment capability.
	\end{Defn}
	
	Beyond the above two drift signals, STEPS also incorporates instantaneous system pressure into the feedback mechanism. In particular, the fulfillment pressure {\small$P_{\mathsf{ful}}^{(t)}$} in~\eqref{eq:fulfillment_pressure} reflects the current contract-violation burden among admitted UDs. However, evaluating only admitted UDs may conceal the pressure induced by highly uncertain requests that are rejected during semantic admission. To capture this effect, we define the semantic admission pressure as
	\begin{equation}
		P_{\mathsf{adm}}^{(t)}
		=
		\begin{cases}
			1-
			\frac{|\mathcal{U}^{\mathsf{s},(t)}|}
			{|\mathcal{U}^{(t)}|},
			&
			|\mathcal{U}^{(t)}|>0,
			\\
			0,
			&
			|\mathcal{U}^{(t)}|=0.
		\end{cases}
		\label{eq:admission_pressure}
	\end{equation}
	Based on the above drift and fulfillment-pressure signals, we define the fulfillment-side feedback signal as
	\begin{equation}
		\Psi_{\mathsf{ful}}^{(t)}
		=
		\alpha_s
		\Delta_{\mathsf{sem}}^{(t)}
		+
		\alpha_f
		\Delta_{\mathsf{ful}}^{(t)}
		+
		\alpha_p
		P_{\mathsf{ful}}^{(t)},
		\label{eq:feedback_signal}
	\end{equation}
	where {\small$\alpha_s\ge0$}, {\small$\alpha_f\ge0$}, and {\small$\alpha_p\ge0$} are weighting coefficients. In \eqref{eq:feedback_signal}, the first two terms capture long-term changes in user semantics and fulfillment capability, whereas the third term reflects the instantaneous fulfillment burden experienced by admitted UDs. Note that {\small$P_{\mathsf{adm}}^{(t)}$} is intentionally excluded from {\small$\Psi_{\mathsf{ful}}^{(t)}$}; instead, it is used separately to regulate semantic admission and prevent excessive rejection of uncertain requests.
	
	The resulting feedback signals are then used to adapt both edge-side resource coordination and semantic contract management. First, after each timeslot, the ESP updates the resource prices for the next timeslot according to the intended resource demand reflected by the game-output action profile before feasibility realization. This allows the pricing mechanism to react to latent resource pressure (i.e., the signed mismatch between the intended aggregate resource demand and the available ES capacity before resource capping). In particular, a positive mismatch indicates that the corresponding ES is over-requested by admitted UDs, even though the final realized allocation is later capped by physical resource limits. Specifically,
	\begin{equation}
		{\small\begin{aligned}
				\lambda_j^{\mathsf{f}, (t+1)}
				&=
				\left[
				\lambda_j^{\mathsf{f},(t)}
				+
				\mu_{\mathsf{f}}
				(1+\chi^{(t)})
				\frac{
					L_j^{\mathsf{f},(t)}(\mathbf{a}^{\mathsf{fin},(t)})
					-
					F_j^{(t)}}
				{F_j^{(t)}}
				\right]_{\lambda_{\mathsf{min}}^{\mathsf{f}}}^{\lambda_{\mathsf{max}}^{\mathsf{f}}},
				\\
				\lambda_j^{\mathsf{b},(t+1)}
				&=
				\left[
				\lambda_j^{\mathsf{b},(t)}
				+
				\mu_{\mathsf{b}}
				(1+\chi^{(t)})
				\frac{
					L_j^{\mathsf{b},(t)}(\mathbf{a}^{\mathsf{fin},(t)})
					-
					B_j^{(t)}}
				{B_j^{(t)}}
				\right]_{\lambda_{\mathsf{min}}^{\mathsf{b}}}^{\lambda_{\mathsf{max}}^{\mathsf{b}}},
		\end{aligned}}
		\label{eq:price_updates}
	\end{equation}
	where {\small$\mu_{\mathsf{f}}$} and {\small$\mu_{\mathsf{b}}$} are price-update stepsizes, and {\small$[\cdot]_{\lambda_{\mathsf{min}}}^{\lambda_{\mathsf{max}}}$} denotes projection onto the corresponding price interval. The updated prices {\small$\lambda_j^{\mathsf{f},(t+1)}$} and {\small$\lambda_j^{\mathsf{b},(t+1)}$} are used as the next-timeslot unit resource prices in the monetary-cost model in~\eqref{eq:monetary_cost}, and consequently affect the contract-guided service loss and best-response scheduling cost through~\eqref{eq:H_loss} and~\eqref{eq:user_cost}. 
	% They do not affect the monetary-cost settlement of the current timeslot. 
	Moreover, a larger {\small$\chi^{(t)}$} increases the sensitivity of the price-update process to congestion, thereby strengthening edge-side coordination under stronger non-stationarity.
	In addition to price adaptation, STEPS updates the semantic admission threshold, contract conservativeness factor, and edge coordination gain according to the observed feedback signals as 
	\begin{equation}
		{\small\begin{aligned}
				\tau^{\rho,{(t+1)}}
				&=
				\Big[
				\tau^{\rho,(t)}\hspace{-1mm}
				-
				\mu_\rho
				\big(\Psi_{\mathsf{ful}}^{(t)}-\Psi_0\big)\hspace{-1mm}
				+
				\mu_a
				\big(P_{\mathsf{adm}}^{(t)}-P_{\mathsf{adm},0}\big)
				\Big]_{\tau_{\mathsf{min}}}^{\tau_{\mathsf{max}}},
				\\
				\beta^{{(t+1)}}
				&=
				\Big[
				\beta^{(t)}
				+
				\mu_\beta
				\big(\Psi_{\mathsf{ful}}^{(t)}-\Psi_0\big)
				\Big]_{\beta_{\mathsf{min}}}^{\beta_{\mathsf{max}}},
				\\
				\chi^{{(t+1)}}
				&=
				\Big[
				\chi^{(t)}
				+
				\mu_\chi
				\big(\Psi_{\mathsf{ful}}^{(t)}-\Psi_0\big)
				\Big]_{\chi_{\mathsf{min}}}^{\chi_{\mathsf{max}}},
		\end{aligned}}
		\label{eq:evolutive_updates}
	\end{equation}
	where {\small$\Psi_0$} denotes the acceptable fulfillment-feedback level and {\small$P_{\mathsf{adm},0}$} denotes the acceptable admission-pressure level. Moreover, {\small$\mu_\rho,~\mu_a,~\mu_\beta,~\mu_\chi\ge0$} are adaptive-control stepsizes, where {\small$\mu_\rho$} controls the sensitivity of the semantic admission threshold to fulfillment-side feedback, {\small$\mu_a$} controls the admission-threshold relaxation driven by admission pressure, {\small$\mu_\beta$} controls the update speed of the contract conservativeness factor, and {\small$\mu_\chi$} controls the update speed of the edge coordination gain. The updated parameters are used in the next timeslot as follows: {\small$\tau^{\rho,(t+1)}$} determines the schedulable UD set through~\eqref{eq:semantic_admission}, {\small$\beta^{(t+1)}$} calibrates the scheduling-stage fulfillment bounds through~\eqref{eq:calibrated_bounds}, and {\small$\chi^{(t+1)}$} enters the best-response scheduling cost and the potential function through~\eqref{eq:user_cost} and~\eqref{eq:potential}. According to \eqref{eq:evolutive_updates}, when {\small$\Psi_{\mathsf{ful}}^{(t)}>\Psi_0$}, the ESP interprets this as evidence of deteriorating fulfillment performance and responds by tightening semantic admission, increasing contract conservativeness, and strengthening edge coordination. When {\small$P_{\mathsf{adm}}^{(t)}>P_{\mathsf{adm},0}$}, the semantic admission threshold is relaxed to avoid excessive rejection of semantically uncertain requests. Consequently, fulfillment pressure and admission pressure influence {\small$\tau^{\rho,(t)}$} in opposite directions, whereas {\small$\beta^{(t)}$} and {\small$\chi^{(t)}$} are driven solely by fulfillment-side feedback.

	\subsection{Overarching Design}
	\label{sec:steps_algorithm}
	
	As explained above, STEPS converts natural language service requests into semantic contracts, performs contract-guided scheduling for admitted UDs, executes the resulting service decisions, and updates scheduling-control parameters using post-execution feedback. In this way, STEPS establishes a closed-loop scheduling framework that adapts to evolving user requirements and system conditions. The overall procedure is summarized in Alg.~\ref{Alg:STEPS}, comprising the steps below.

	\begin{algorithm}[t]
		{\scriptsize \setstretch{0.99}
			\caption{Proposed STEPS}
			\label{Alg:STEPS}
			\LinesNumbered
			
			{\color{black}
				\textbf{Input:} Task attributes $\{q_i^{(t)}\}$,
				SEDs $\{\mathcal{L}_i^{(t)}\}$, system context
				$\xi^{(t)}$, ES resource prices,
				adaptive-control parameters
				$\tau^{\rho,(t)},\beta^{(t)},\chi^{(t)}$,
				maximum round number $I_{\mathsf{max}}$,
				and improvement tolerance $\varepsilon_{\mathsf{br}}$.
			}
			
			\For{$t=1$ \KwTo $T$}{
				
				\For{$u_i\in\mathcal{U}^{(t)}$}{
					
					{\color{black}
						Parse $\mathcal{L}_i^{(t)}$ into semantic
						levels and raw confidence, then calibrate
						the parsing confidence\;
					}
					
					{\color{black}
						Generate preference weights and original
						fulfillment bounds by
						\eqref{eq:weight_mapping}--\eqref{eq:bound_mapping}\;
					}
					
					{\color{black}
						Construct the individually feasible action
						set $\mathcal{A}_i^{\mathsf{ind},(t)}$\;
					}
					
					{\color{black}
						Evaluate context-inconsistency risk
						by~\eqref{eq:context_risk} and semantic
						uncertainty by~\eqref{eq:semantic_uncertainty}\;
					}
					
					{\color{black}
						Generate semantic contract
						$\mathcal{C}_i^{(t)}
						=(\mathbf{w}_i^{(t)},\mathbf{g}_i^{(t)},\rho_i^{(t)})$\;
					}
				}
				
				{\color{black}
					Construct the schedulable UD set
					$\mathcal{U}^{\mathsf{s},(t)}$
					by~\eqref{eq:semantic_admission}\;
				}
				
				\For{$u_i\in\mathcal{U}^{\mathsf{s},(t)}$}{
					
					{\color{black}
						Calibrate original fulfillment bounds
						by~\eqref{eq:calibrated_bounds}\;
					}
					
					{\color{black}
						Set the scheduling action space
						$\mathcal{A}_i^{(t)}
						\leftarrow\mathcal{A}_i^{\mathsf{ind},(t)}$\;
					}
				}
				
				Initialize action profile
				$\mathbf{a}^{(t)}(0)$
				and set $r\leftarrow0$\;
				
				\Repeat{
					$\mathsf{improved}=\mathsf{false}$
					\textbf{or} $r\ge I_{\mathsf{max}}$
				}{
					
					Set
					$\mathsf{improved}\leftarrow\mathsf{false}$\;
					
					{\color{black}
						Visit admitted UDs sequentially\;
					}
					
					\For{each visited UD $u_i$}{
						
						Evaluate $\mathcal{J}_i^{(t)}$
						in~\eqref{eq:user_cost}
						for all
						$a_i^{(t)}\in\mathcal{A}_i^{(t)}$\;
						
						$a_i^{\mathsf{br},(t)}
						\leftarrow
						\arg\min_{a_i^{(t)}\in\mathcal{A}_i^{(t)}}
						\mathcal{J}_i^{(t)}
						(a_i^{(t)},\mathbf{a}_{-i}^{(t)})$\;
						
						{\color{black}
							\If{
								the best response reduces
								$\mathcal{J}_i^{(t)}$
								by more than
								$\varepsilon_{\mathsf{br}}$
							}{
								Update
								$a_i^{(t)}
								\leftarrow a_i^{\mathsf{br},(t)}$\;
								
								Set
								$\mathsf{improved}
								\leftarrow\mathsf{true}$\;
								
								Update ES loads and congestion
								penalties by
								\eqref{eq:intended_compute_load}--%
								\eqref{eq:congestion_penalty}\;
							}
						}
					}
					
					$r\leftarrow r+1$\;
				}
				
				Set
				$\mathbf{a}^{\mathsf{fin},(t)}
				\leftarrow\mathbf{a}^{(t)}(r)$
				as the final intended action profile\;
				
				{\color{black}
					Realize capacity-feasible resource allocations
					by
					\eqref{eq:compute_realization_factor}--%
					\eqref{eq:realized_resource_allocation}\;
				}
				
				Execute scheduled tasks and observe
				service outcomes\;
				
				{\color{black}
					Compute fulfillment deviation
					by~\eqref{eq:ful_deviation},
					fulfillment pressure
					by~\eqref{eq:fulfillment_pressure},
					drifts by
					\eqref{eq:semantic_drift}--%
					\eqref{eq:fulfillment_drift},
					admission pressure
					by~\eqref{eq:admission_pressure},
					and fulfillment-side feedback
					by~\eqref{eq:feedback_signal}\;
				}
				
				{\color{black}
					Update ES resource prices for the next
					timeslot by~\eqref{eq:price_updates}\;
				}
				
				Update
				$\tau^{\rho,(t+1)}$,
				$\beta^{(t+1)}$,
				and $\chi^{(t+1)}$
				by~\eqref{eq:evolutive_updates}\;
			}
			
			\textbf{Return:}
			$\{\mathbf{x}^{(t)},\mathbf{f}^{(t)},\mathbf{b}^{(t)}\}$,
			updated prices, and updated adaptive parameters.
		}
	\end{algorithm}

	{\color{black}
		\noindent\textbf{Step 1. Semantic Contract Generation and Context Risk:}
		At the beginning of each timeslot, each UD submits its task attributes and associated SED. The LLM-assisted semantic parser extracts semantic service levels and raw parsing confidence, which is subsequently calibrated. The semantic levels are mapped into preference weights and original fulfillment bounds. STEPS then constructs the individually feasible action set $\mathcal{A}_i^{\mathsf{ind},(t)}$ to evaluate the context-inconsistency risk by~\eqref{eq:context_risk}. The calibrated confidence and context-inconsistency risk jointly determine semantic uncertainty, yielding the semantic contract $\mathcal{C}_i^{(t)}$.
	}
	
	{\color{black}
		\noindent\textbf{Step 2. Semantic Admission and Contract Calibration:}
		The ESP constructs the schedulable UD set $\mathcal{U}^{\mathsf{s},(t)}$ according to the semantic uncertainty threshold $\tau^{\rho,(t)}$. For each admitted UD, STEPS calibrates the original fulfillment bounds using $\rho_i^{(t)}$ and the conservativeness factor $\beta^{(t)}$, and adopts the individually feasible action set as the finite scheduling action space $\mathcal{A}_i^{(t)}$.
	}
	
	{\color{black}
		\noindent\textbf{Step 3. Potential Game and Asynchronous Best Response:}
		STEPS applies ESP-coordinated asynchronous best-response updates to the per-slot surrogate problem $\mathcal{P}_{\mathsf{pg}}^{(t)}$ in~\eqref{eq:pg_surrogate_problem}. Starting from an initial action profile, each admitted UD evaluates its feasible actions and updates its decision only when the scheduling-cost reduction exceeds $\varepsilon_{\mathsf{br}}$. By the exact-potential property, each accepted update also decreases $\Phi^{(t)}$. The process terminates when a complete round produces no accepted improvement or the maximum number of rounds $I_{\mathsf{max}}$ is reached, yielding the intended EN selections and resource-allocation requests.
	}
	
	{\color{black}
		\noindent\textbf{Step 4. Feasibility Realization and Execution:}}
	The intended decisions are converted into physically feasible resource allocations through the realization mechanism. {\color{black}Resulting allocations are used to execute the tasks and evaluate their delay, energy consumption, monetary cost, and T4E outcomes.}
	
	{\color{black}
		\noindent\textbf{Step 5. Execution Feedback:}}
	The ESP computes the fulfillment deviation {\small$\Omega_i^{\mathsf{ful},(t)}$} and fulfillment degree {\small$\varphi_i^{(t)}$} for each admitted UD. It then evaluates semantic request drift, fulfillment drift, fulfillment pressure, and admission pressure.
	
	{\color{black}
		\noindent\textbf{Step 6. Adaptive Scheduling-Control Update:}}
	STEPS finally closes the feedback loop by updating the scheduling-control parameters for the next timeslot. {\color{black}Resource prices are adjusted according to the intended resource demand, while the contract conservativeness factor and edge coordination gain are updated using the fulfillment-side feedback signal. The semantic admission threshold is further adjusted according to both fulfillment-side feedback and admission pressure.}

	{\color{black}
		\subsection{Computational Complexity of STEPS}
		\label{app:complexity}
		
		Let $N=|\mathcal{U}^{(t)}|$, $N_s=|\mathcal{U}^{\mathsf{s},(t)}|$, and $M=|\mathcal{E}|$ denote the numbers of incoming UDs, admitted UDs, and ESs, respectively. Let $K=\max_{u_i\in\mathcal{U}^{(t)}}|\mathcal{A}_i^{\mathsf{ind},(t)}|$ denote the maximum number of candidate actions per UD, where $K=\mathcal{O}(M)$ for a fixed resource-package configuration. Constructing the feasible action sets and evaluating context-inconsistency risk require $\mathcal{O}(NK)$ operations, while the remaining semantic mapping and admission procedures incur $\mathcal{O}(N)$ complexity, excluding LLM inference. During asynchronous best-response scheduling, each admitted UD evaluates at most $K$ candidate actions per round. With $R\le I_{\mathsf{max}}$ update rounds, incrementally maintained ES loads, and constant-time candidate evaluation, the overall per-slot computational complexity is
		\begin{equation}
			\mathcal{O}\left(NK+RN_sK+M\right),
			\label{eq:steps_complexity}
		\end{equation}
		where the final term accounts for resource realization and control updates. Thus, STEPS achieves polynomial per-slot computational complexity under the prescribed iteration limit.
		{\color{black}
			Moreover, if the ESP exchanges aggregate load information of candidate ESs with UDs during best-response updates, the signaling overhead is bounded by $\mathcal{O}(RN_sM)$ per timeslot.
		}
	}

	{\color{black}
		\subsection{Key Properties}
		\label{subsec:properties}
		
		We next establish key properties of STEPS. Recall that the original long-term problem $\mathcal{P}$ is addressed by the per-slot penalized surrogate problem $\mathcal{P}_{\mathsf{pg}}^{(t)}$ in~\eqref{eq:pg_surrogate_problem}. For each timeslot $t$, the admitted UD set, individually feasible action sets, semantic contracts, pre-decision system states, resource prices, and adaptive control parameters remain fixed during $\mathcal{G}^{(t)}$. The interaction among UDs arises solely through the ES congestion penalties. Accordingly, the following equilibrium and convergence results characterize the intended scheduling game before capacity-feasible resource realization.
		
		\begin{thm}[Exact Potential Property of STEPS]
			\label{thm:exact_potential}
			For each timeslot $t$, the contract-guided scheduling game $\mathcal{G}^{(t)}$ is an exact potential game with potential function $\Phi^{(t)}$ defined in~\eqref{eq:potential}. Specifically, for any admitted UD $u_i\in\mathcal{U}^{\mathsf{s},(t)}$, any fixed action profile $\mathbf{a}_{-i}^{(t)}$ of the other UDs, and any two feasible actions $a_i^{(t)},\widehat{a}_i^{(t)}\in\mathcal{A}_i^{(t)}$, we have $			\mathcal{J}_i^{(t)}(\widehat{a}_i^{(t)},\mathbf{a}_{-i}^{(t)})
			\hspace{-.5mm}-\hspace{-.5mm}\mathcal{J}_i^{(t)}(a_i^{(t)},\mathbf{a}_{-i}^{(t)})\hspace{-.5mm}=\hspace{-.5mm}\Phi^{(t)}(\widehat{a}_i^{(t)},\mathbf{a}_{-i}^{(t)})
			-\Phi^{(t)}(a_i^{(t)},\mathbf{a}_{-i}^{(t)}).
			$\end{thm}
		
		\begin{proof}
			Fix an arbitrary timeslot $t$, an admitted UD $u_i$, and the actions $\mathbf{a}_{-i}^{(t)}$ of the other UDs. For notational simplicity, we omit the timeslot superscript from the action arguments within this proof. Consider a unilateral deviation of UD $u_i$ from $a_i$ to $\widehat{a}_i$. Under the fixed-per-slot setting, the contract-guided service loss $H_k^{(t)}(a_k)$ in~\eqref{eq:H_loss} depends only on UD $u_k$'s own action. Hence, the change in the aggregate service-loss term of~\eqref{eq:potential} is $H_i^{(t)}(\widehat{a}_i)-H_i^{(t)}(a_i)$.
			
			For each resource type $\mathsf{z}\in\{\mathsf{f},\mathsf{b}\}$, the aggregate ES load contributed by the other UDs is fixed and is given by $L_{j,-i}^{\mathsf{z},(t)}(\mathbf{a}_{-i})$ in~\eqref{eq:load_excluding_i}. By the definition of the marginal congestion increment $\Delta\Gamma_i^{\mathsf{z},(t)}$, the aggregate congestion potential under action $a_i$ equals the congestion potential generated by $\mathbf{a}_{-i}$ plus $\Delta\Gamma_i^{\mathsf{z},(t)}(a_i,\mathbf{a}_{-i})$. The congestion contribution of $\mathbf{a}_{-i}$ is unchanged by the unilateral deviation and therefore cancels when the two action profiles are compared.
			Combining the changes in service loss and congestion potential yields
			\begin{equation}
				{\small\begin{aligned}
						&\Phi^{(t)}(\widehat{a}_i,\mathbf{a}_{-i})
						-\Phi^{(t)}(a_i,\mathbf{a}_{-i})=H_i^{(t)}(\widehat{a}_i)-H_i^{(t)}(a_i)\\
						&+(1+\chi^{(t)})
						\sum_{\mathsf{z}\in\{\mathsf{f},\mathsf{b}\}}
						\Big[
						\Delta\Gamma_i^{\mathsf{z},(t)}
						(\widehat{a}_i,\mathbf{a}_{-i})-
						\Delta\Gamma_i^{\mathsf{z},(t)}
						(a_i,\mathbf{a}_{-i})
						\Big]\\
						&=\mathcal{J}_i^{(t)}(\widehat{a}_i,\mathbf{a}_{-i})
						-\mathcal{J}_i^{(t)}(a_i,\mathbf{a}_{-i}),
				\end{aligned}}
			\end{equation}
			where the last equality follows from the scheduling cost in~\eqref{eq:user_cost}. Thus, every unilateral change in a UD's scheduling cost equals the corresponding change in $\Phi^{(t)}$, establishing the exact-potential property.
		\end{proof}
		
		\begin{thm}[Existence of Pure-Strategy Nash Equilibrium]
			\label{thm:pne_existence}
			For each timeslot $t$, the contract-guided scheduling game $\mathcal{G}^{(t)}$ admits at least one pure-strategy Nash equilibrium.
		\end{thm}
		
		\begin{proof}
			Since each admitted UD has a nonempty and finite action space $\mathcal{A}_i^{(t)}$, the joint action space $\mathcal{A}^{(t)}=\prod_{u_i\in\mathcal{U}^{\mathsf{s},(t)}}\mathcal{A}_i^{(t)}$ is also finite and nonempty.
			By Theorem~\ref{thm:exact_potential}, the game admits the exact potential function $\Phi^{(t)}$, which attains a minimum over $\mathcal{A}^{(t)}$. Let $\mathbf{a}^{\star,(t)}$ be a minimizing action profile, and suppose that it is not a pure-strategy Nash equilibrium. Then some admitted UD can strictly reduce its scheduling cost through a feasible unilateral deviation. By the exact-potential property, the same deviation would strictly decrease $\Phi^{(t)}$, contradicting the minimality of $\mathbf{a}^{\star,(t)}$. Therefore, no admitted UD has a strictly improving unilateral deviation, and $\mathbf{a}^{\star,(t)}$ is a pure-strategy Nash equilibrium. If no UD is admitted, the empty action profile is trivially an equilibrium.
		\end{proof}
		
		\begin{thm}[Finite-Step Convergence Under Strict Best Response]
			\label{thm:finite_convergence}
			For each timeslot $t$, consider the asynchronous best-response process in~\eqref{eq:best_response_update} without the iteration cap $I_{\mathsf{max}}$ or a positive numerical improvement tolerance. Suppose that every admitted UD is visited in each update round and adopts a best response whenever it strictly reduces its scheduling cost. Then the process terminates after finitely many effective updates at a pure-strategy Nash equilibrium of $\mathcal{G}^{(t)}$.
		\end{thm}
		
		\begin{proof}
			Consider an arbitrary timeslot $t$. Under the strict best-response rule, every effective action update strictly decreases the scheduling cost of the updating UD. By Theorem~\ref{thm:exact_potential}, the same update produces an equal strict decrease in the potential function $\Phi^{(t)}$. Since the joint action space $\mathcal{A}^{(t)}$ is finite, the potential function takes only finitely many values over all feasible action profiles. Moreover, no action profile can be revisited after a strict potential decrease. Thus, an infinite sequence of effective updates is impossible, and the process must terminate after finitely many such updates.
			At termination, a complete update round produces no strictly improving unilateral deviation. Since every admitted UD evaluates its feasible best response, no UD can further reduce its scheduling cost by changing its action alone. Hence, the resulting action profile is a pure-strategy Nash equilibrium of $\mathcal{G}^{(t)}$.
		\end{proof}
		
		In numerical implementation, STEPS limits the update rounds to $I_{\mathsf{max}}=20$ and accepts an action update only when the scheduling-cost reduction exceeds $\varepsilon_{\mathsf{br}}=10^{-9}$. To assess the remaining unilateral improvement of the returned intended action profile $\mathbf{a}^{\mathsf{fin},(t)}$, we define the Nash-equilibrium gap as
		\begin{equation}
			{\small\begin{aligned}
					\epsilon_{\mathsf{NE}}^{(t)}
					=\max_{u_i\in\mathcal{U}^{\mathsf{s},(t)}}\Big[
					\mathcal{J}_i^{(t)}
					(a_i^{\mathsf{fin},(t)},\mathbf{a}_{-i}^{\mathsf{fin},(t)})-\min_{\bar a_i\in\mathcal{A}_i^{(t)}}
					\mathcal{J}_i^{(t)}
					(\bar a_i,\mathbf{a}_{-i}^{\mathsf{fin},(t)})
					\Big]_+.
			\end{aligned}}
			\label{eq:ne_gap}
		\end{equation}
		For a timeslot without admitted UDs, we set $\epsilon_{\mathsf{NE}}^{(t)}=0$. If the numerical process terminates after a complete update round without an accepted improvement, then $\epsilon_{\mathsf{NE}}^{(t)}\le\varepsilon_{\mathsf{br}}$. In contrast, termination solely upon reaching $I_{\mathsf{max}}$ does not guarantee this bound. An exact Nash equilibrium corresponds to $\epsilon_{\mathsf{NE}}^{(t)}=0$. These equilibrium properties concern the intended action profile before resource realization and do not necessarily extend to the final proportionally scaled allocations.
		
		\begin{thm}[Boundedness of Adaptive Control Parameters]
			\label{thm:bounded_updates}
			Suppose that the adaptive-control parameters and resource prices are initialized within their prescribed projection intervals. Then, for every timeslot $t\in\mathcal{T}$, the adaptive-control parameters satisfy
			$
			\tau^{\rho,(t)}\in[\tau_{\mathsf{min}},\tau_{\mathsf{max}}],
			\beta^{(t)}\in[\beta_{\mathsf{min}},\beta_{\mathsf{max}}]$, and $
			\chi^{(t)}\in[\chi_{\mathsf{min}},\chi_{\mathsf{max}}].
			$
			Moreover, the resource prices satisfy
			$
			\lambda_j^{\mathsf{z},(t)}
			\in[\lambda_{\mathsf{min}}^{\mathsf{z}},
			\lambda_{\mathsf{max}}^{\mathsf{z}}],
			\quad e_j\in\mathcal{E},\ 
			\mathsf{z}\in\{\mathsf{f},\mathsf{b}\}.
			$
		\end{thm}
		
		\begin{proof}
			The adaptive-control updates in~\eqref{eq:evolutive_updates} employ the projection operator $[x]_a^b=\min\{\max\{x,a\},b\}$, which maps any input to  $[a,b]$. Since the adaptive-control parameters are initialized within their prescribed intervals, each projected update keeps $\tau^{\rho,(t)}$, $\beta^{(t)}$, and $\chi^{(t)}$ within the corresponding bounds. Likewise, the projected resource-price updates in~\eqref{eq:price_updates} preserve the prescribed price intervals. The stated boundedness properties therefore hold for all timeslots.
		\end{proof}
		
		Theorem~\ref{thm:bounded_updates} guarantees the boundedness of the projected adaptive-control parameters and resource prices throughout the considered time horizon, without implying convergence of the feedback dynamics.
	}
	\section{Numerical Evaluations}
	\label{sec:evaluation}

	\subsection{Experimental Setting}
	\label{sec:experimental_settings}
	
	{\color{black}
		
		We evaluate STEPS under synthetic edge-network settings and a Melbourne-based hybrid setting, with results averaged over $30$ paired Monte Carlo trials unless otherwise specified. {\color{black}Error bars and shaded regions, where shown, indicate $95\%$ confidence intervals.} Natural-language requests are generated from semantic templates covering delay, energy, monetary cost, and T4E requirements, while ambiguous, conflicting, and out-of-scope requests are injected with ratio $r_{\mathsf{unc}}\in[0,0.5]$ to emulate semantic uncertainty. For non-stationary evaluation, a joint drift event is introduced at timeslot $t_{\mathsf d}=50$ with intensity $\eta_{\mathsf d}=0.6$, involving changes in semantic preferences, task workloads, and edge resource conditions.
	}{\color{black}
		Semantic parsing is performed by Qwen3.5-4B served locally through Ollama, with the parsing outputs cached in the main scheduling experiments for reproducible comparisons. Raw parsing confidence is calibrated using an independent set of $120$ labeled natural-language requests {\color{black}constructed specifically for this study (see Appx.~\ref{app:semantic_validation} for the dataset construction)} before semantic contract generation.
	}All experiments are conducted on a workstation equipped with a 12th Gen Intel Core i9-12900H CPU and an NVIDIA GeForce RTX 3060 GPU. Following representative studies~\cite{xiao2024adaptive,li2024optimal,sun2026igaa,wang2025intent}, key simulation parameters are summarized in Table~\ref{tab:main_settings}.

	\begin{table}[t]
		\centering
		\caption{Simulation settings.}
		\label{tab:main_settings}
		
		\setlength{\tabcolsep}{1.0mm}
		\renewcommand{\arraystretch}{1.10}
		\begin{tabular}{
				>{\raggedright\arraybackslash\scriptsize}m{0.27\columnwidth}|
				>{\raggedright\arraybackslash\scriptsize}m{0.65\columnwidth}}
			\hline
			\rowcolor[gray]{0.9}
			\textbf{Aspect}
			& \textbf{Setting} \\
			\hline
			
			Network Topology
			& {\color{black}
				$800\times800~\mathsf{m}^2$;
				S1--S4 (UDs, ESs): $(50,2)$,
				$(100,4)$, $(200,8)$, $(300,12)$.
			} \\
			\hline
			
			{\color{black}Wireless, Computing, and Tasks}
			&
			$N_0=-174~\mathsf{dBm/Hz}$;
			BW $20$--$40~\mathsf{MHz}$;
			Tx power $10$--$200~\mathsf{mW}$;
			local/edge CPU $0.5$--$2.0$ /
			$20$--$80~\mathsf{GHz}$;
			data size $5$--$20~\mathsf{MB}$;
			workload $0.5$--$5.0\times10^9$ cycles;
			T4E score $0.4$--$1.0$.
			\\
			\hline
			
			{\color{black}Semantic Requests and Parser}
			&
			Delay, energy, monetary cost, and T4E;
			three levels (low/medium/high);
			$r_{\mathsf{unc}}\in[0,0.5]$;
			Qwen3.5-4B via Ollama;
			cached template-level parsing.
			\\
			\hline
			
			{\color{black}Semantic Uncertainty}
			& {\color{black}
				$\alpha_c=\alpha_x=0.5$;
				$\omega_{\mathsf{k}}=0.25$;
				$\theta_w=2$;
				nearest-ES reference with
				CPU/BW shares $0.08$.
			} \\
			\hline
			
			STEPS Configuration
			& {\color{black}
				$\tau_0^\rho=0.45$,
				$\tau^\rho\in[0.20,0.70]$;
				$\beta_0=0.5$, $\beta\in[0,1]$;
				$\chi_0=1.5$, $\chi\in[0,3]$;
				$\gamma_{\mathsf f}=\gamma_{\mathsf b}=10$;
				$\varrho_{\mathsf f}=\varrho_{\mathsf b}=1$.
			} \\
			\hline
			
			{\color{black}Contract Penalties}
			& {\color{black}
				$\zeta_{\mathsf D}=\zeta_{\mathsf E}
				=\zeta_{\mathsf M}=\zeta_{\mathsf S}=1$;
				$\eta=1$.
			} \\
			\hline
			
			{\color{black}Feedback Configuration}
			& {\color{black}
				$\mu_\rho=\mu_a=\mu_\beta=0.05$;
				$\mu_\chi=0.10$;
				$\mu_{\mathsf f}=\mu_{\mathsf b}=0.005$;
				$\alpha_s=\alpha_f=\alpha_p=1/3$;
				$\Psi_0=P_{\mathsf{adm},0}=0.10$;
				recent/historical windows $5/10$ slots.
			} \\
			\hline

			{\color{black}Prices and Resource Packages}
			& {\color{black}
				Initial compute/BW prices:
				$[0.03,0.09]$ MU/GHz and
				$[0.02,0.08]$ MU/MHz;
				resource allocations expressed in
				corresponding pricing units;
				price projection $[0.01,0.20]$;
				local CPU shares $\{0.50,0.75,1.00\}$;
				edge CPU/BW shares $\{0.04,0.08,0.14\}$.
			} \\
			\hline

			{\color{black}Drift and Evaluation}
			& {\color{black}
				$t_{\mathsf d}=50$, $\eta_{\mathsf d}=0.6$;
				$T=100$; $30$ paired trials;
				$I_{\mathsf{max}}=20$;
				$\varepsilon_{\mathsf{br}}=10^{-9}$;
				$\epsilon_{\mathsf v}=0.10$.
			} \\
			\hline
			
		\end{tabular}
	\end{table}
	
	\subsection{Benchmark Methods}
	\label{sec:benchmarks}
	
	We compare STEPS with five representative benchmarks:
	{\color{black}
		\textit{(i) Param-Opt} is a parameter-driven optimization baseline adapted from conventional edge AI resource scheduling studies~\cite{li2024optimal}. It assumes direct access to numerical service requirements derived from the reference request templates and employs the same best-response scheduling mechanism as STEPS.
		\textit{(ii) Direct-LLM} is an LLM-enabled intent-translation baseline adapted from service configuration methods~\cite{mekrache2024llm}. It uses the same LLM parsing outputs to construct scheduling requirements without explicitly accounting for semantic uncertainty or conservative contract calibration.
		\textit{(iii) DRL-Opt} is a learning-based scheduling baseline adapted from edge resource optimization studies~\cite{li2024optimal}. It employs a Double-DQN policy with feasibility masking to select execution nodes and resource configurations, followed by the common capacity-feasible resource-realization procedure.
		\textit{(iv) Intent-Reconf} is a violation-triggered intent management baseline adapted from intent-driven resource management studies~\cite{akbari2025intentcontinuum}. It monitors service violations and triggers corrective resource reconfiguration upon subsequent requests from the same user, rather than performing uncertainty-aware contract calibration before scheduling.
		\textit{(v) GAH} is a generalized-assignment heuristic baseline adapted from edge LLM scheduling studies~\cite{zhang2025beyond}. It greedily selects execution nodes and resource configurations according to task urgency and residual resource capacities.
	}

	We also evaluate two ablation variants for STEPS:
	\textit{NoCal} removes semantic admission and uncertainty-aware calibration by setting {\small$\tau^{\rho,(t)}=1$} and {\small$\beta^{(t)}=0$};
	{\color{black}
		\textit{NoEvo} disables fulfillment-driven evolution by fixing $\tau^{\rho,(t)}=\tau_0^\rho=0.45$, $\beta^{(t)}=\beta_0$, $\chi^{(t)}=\chi_0$, and the resource prices at their initial values, while retaining semantic admission and uncertainty-aware contract calibration.
	}

	\subsection{Evaluation Protocol and Performance Metrics}
	\label{sec:metrics}
	
	For fair comparison, all methods are evaluated under the same network topology, task workloads, wireless conditions, edge resource states, and request templates. Fulfillment-related metrics are computed using the same template-derived semantic contracts and task attributes, rather than method-dependent internal bounds, to avoid biased self-assessment.
	
	Let {\small$\mathcal{I}_m=\{(i,t)\mid u_i \text{ is served under method } m \text{ at timeslot } t\}$} denote the served user--timeslot pairs of method $m$. We consider the following performance metrics (for methods without semantic admission, all UDs are treated as admitted):
	
	\noindent\textit{1) Semantic Contract Fulfillment and Admission:}
	{\color{black}
		We use average fulfillment degree (AFD), overall average fulfillment degree (OAFD), contract violation rate (CVR), average contract-guided service loss (ACSL), and semantic admission ratio (SAR) to evaluate contract fulfillment and semantic admission behavior.
	}
	Specifically, {\small$\mathsf{AFD}_m=\frac{1}{|\mathcal{I}_m|}\sum_{(i,t)\in\mathcal{I}_m}\varphi_i^{(t)}$} measures the average semantic contract fulfillment degree, where a larger value indicates better fulfillment.
	{\color{black}
		To account for both fulfillment quality and service coverage, we further define the overall average fulfillment degree as {\small$\mathsf{OAFD}_m=\frac{\sum_{(i,t)\in\mathcal{I}_m}\varphi_i^{(t)}}{\sum_{t=1}^{T}|\mathcal{U}^{(t)}|}$}, where unserved requests contribute zero fulfillment.
	}
	To avoid over-counting negligible deviations, CVR is defined as {\small$\mathsf{CVR}_m=\frac{1}{|\mathcal{I}_m|}\sum_{(i,t)\in\mathcal{I}_m}\mathbb{I}_{\{\Omega_i^{\mathsf{ful},(t)}>\epsilon_{\mathsf v}\}}$}, where {\small$\epsilon_{\mathsf v}=0.10$} unless otherwise specified. A smaller CVR indicates fewer significant contract violations. Since CVR only measures the frequency of threshold-exceeding fulfillment deviations, rather than their magnitude or the overall service quality, it is interpreted jointly with ACSL. Specifically, {\small$\mathsf{ACSL}_m=\frac{1}{|\mathcal{I}_m|}\sum_{(i,t)\in\mathcal{I}_m}J_{i,\mathsf{eval}}^{(t)}$} {\color{black}measures the contract-guided service loss evaluated according to~\eqref{eq:J_loss}, using shared reference preference weights and original fulfillment bounds for all methods.}
	Since STEPS includes semantic admission, we further report {\small$\mathsf{SAR}_m=\frac{1}{T}\sum_{t=1}^{T}\frac{|\mathcal{U}_m^{\mathsf{s},(t)}|}{|\mathcal{U}^{(t)}|}$} to quantify the admission cost associated with fulfillment improvement.
	
	\noindent\textit{2) Service Efficiency and Online Scalability:}
	We use average service delay (ASD) and running time (RT) to evaluate service efficiency and online scheduling overhead. ASD is defined as {\small$\mathsf{ASD}_m=\frac{1}{|\mathcal{I}_m|}\sum_{(i,t)\in\mathcal{I}_m}D_i^{(t)}$}, where {\small$D_i^{(t)}$} is the task service delay in~\eqref{eq:delay} under the final feasible allocation.
	{\color{black}
		RT denotes the average computational time of the per-slot scheduling procedure, excluding semantic parsing, contract construction, admission, and subsequent resource realization.
	}
	
	\noindent\textit{3) Non-Stationary Drift Adaptation:}
	{\color{black}
		To evaluate scheduling performance under non-stationary conditions, we introduce joint changes in semantic preferences, task workloads, and network resource conditions at timeslot $t_{\mathsf d}$. We report per-slot OAFD, CVR, ACSL, and SAR to characterize the responses of different methods to these changes.
	}

	\subsection{Synthetic Experiments}
	\label{sec:synthetic_experiments}

	\begin{figure*}[t]
		\centering
		\setlength{\abovecaptionskip}{-1mm}
		\includegraphics[width=0.98\textwidth]{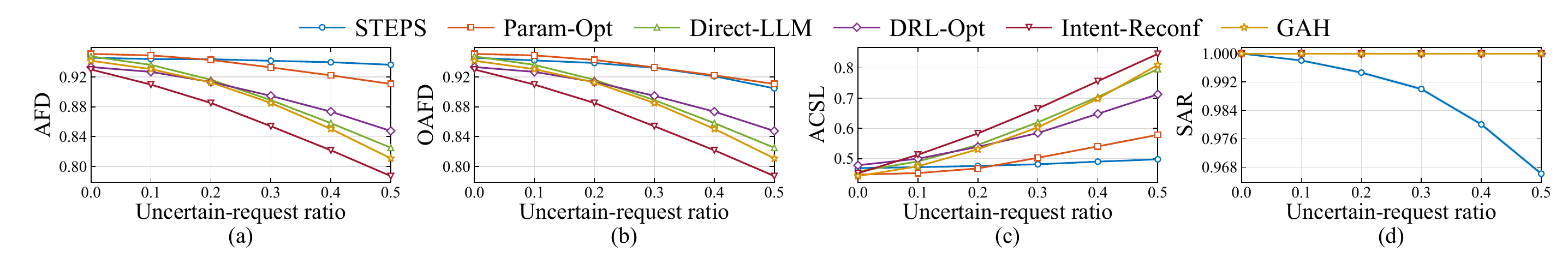}
		\caption{{\color{black}
				Impact of semantic uncertainty on various performance metrics:
				(a) AFD, (b) OAFD, (c) ACSL, and (d) SAR.
		}}
		\label{fig:semantic_uncertainty}
		
	\end{figure*}

\subsubsection{Impact of Semantic Uncertainty}

We first evaluate the robustness of different methods under increasing semantic uncertainty.
{\color{black}
	Fig.~\ref{fig:semantic_uncertainty}(a) shows that STEPS experiences only a slight decline in AFD as $r_{\mathsf{unc}}$ increases, whereas Direct-LLM, DRL-Opt, Intent-Reconf, and GAH exhibit more pronounced degradation. This resilience is attributed to the joint consideration of calibrated parsing confidence and context-inconsistency risk, which guides semantic admission by identifying requests with higher semantic uncertainty. For admitted UDs, uncertainty-aware contract calibration further adjusts the fulfillment bounds to account for residual uncertainty. In comparison, Param-Opt maintains relatively high AFD by directly accessing numerical reference requirements. At higher uncertainty levels, STEPS achieves higher AFD than Param-Opt, demonstrating the effectiveness of semantic contract-guided scheduling. As shown in Fig.~\ref{fig:semantic_uncertainty}(b), STEPS also maintains OAFD close to Param-Opt, while its advantage over the four practical benchmarks becomes increasingly evident as semantic uncertainty increases. The difference between AFD and OAFD reflects the effect of semantic admission on overall service coverage.
	Fig.~\ref{fig:semantic_uncertainty}(c) shows that the ACSL of STEPS increases moderately with $r_{\mathsf{unc}}$, while the four practical benchmarks experience substantially greater service loss as uncertainty increases. Although some benchmarks achieve comparable or lower ACSL under low uncertainty, STEPS becomes more favorable as uncertain requests become more prevalent. This improvement is supported by uncertainty-aware contract calibration and congestion-aware best-response scheduling, which jointly consider multidimensional service preferences, contract violations, and shared ES resource contention. Meanwhile, Fig.~\ref{fig:semantic_uncertainty}(d) shows that the SAR of STEPS gradually decreases as more uncertain requests are excluded from scheduling, whereas the benchmarks without semantic admission maintain an SAR of one. Overall, STEPS maintains favorable contract fulfillment and relatively low service loss under increasing semantic uncertainty, with a modest reduction in admission coverage.
}

\subsubsection{Scalability Under Different Problem Scales}

\begin{figure*}[t]
	\centering
	\setlength{\abovecaptionskip}{-1mm}
	\includegraphics[width=0.98\textwidth]{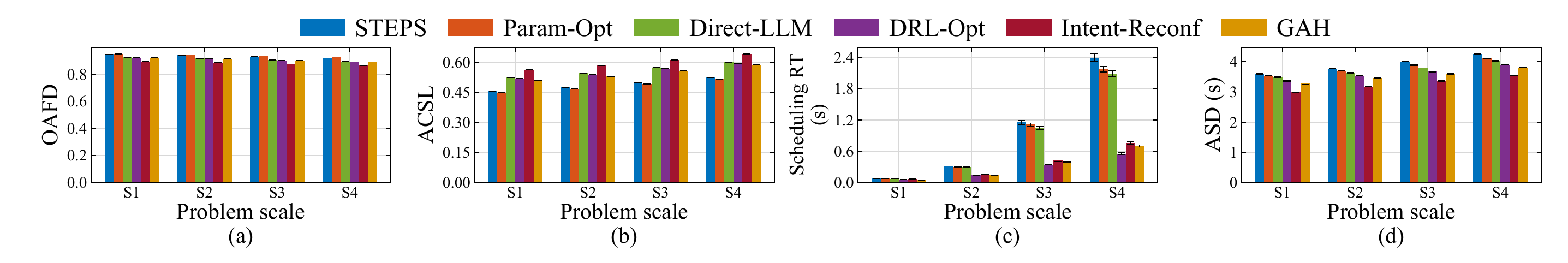}
	\caption{{\color{black}Scalability under different problem scales: (a) OAFD, (b) ACSL, (c) RT, and (d) ASD.}}
	\label{fig:problem_scale}
\end{figure*}

We next evaluate the scalability of STEPS under different problem scales. Specifically, the four scales are defined as $\mathrm{S1}: (|\mathcal{U}^{(t)}|,|\mathcal{E}|)=(50,2)$, $\mathrm{S2}: (100,4)$, $\mathrm{S3}: (200,8)$, and $\mathrm{S4}: (300,12)$. These settings jointly increase the numbers of UDs and ESs, leading to progressively larger scheduling instances.
{\color{black}
	Figs.~\ref{fig:problem_scale}(a) and (b) show that STEPS maintains higher OAFD and lower ACSL than Direct-LLM, DRL-Opt, Intent-Reconf, and GAH across the four problem scales, while remaining competitive with Param-Opt. The latter achieves slightly higher OAFD and lower ACSL. This advantage stems from contract-guided resource allocation and congestion-aware best-response updates, which jointly account for multidimensional service requirements and shared ES resource contention as the scheduling scale increases.
	Fig.~\ref{fig:problem_scale}(c) shows that RT increases with the problem scale, with STEPS incurring higher computational overhead due to its iterative evaluation of EN selections and resource-package combinations. Fig.~\ref{fig:problem_scale}(d) further shows that STEPS generally incurs higher ASD, especially at larger scales, reflecting the tradeoff between service delay and the other contract dimensions, including energy consumption, monetary cost, and T4E. Overall, STEPS maintains strong contract fulfillment and low service loss across different problem scales, at the cost of additional scheduling overhead and service delay.
}

\subsubsection{Adaptation Under Non-Stationary Drift}

\begin{figure*}[t]
	\centering
	\setlength{\abovecaptionskip}{-1mm}
	\includegraphics[width=0.98\textwidth]{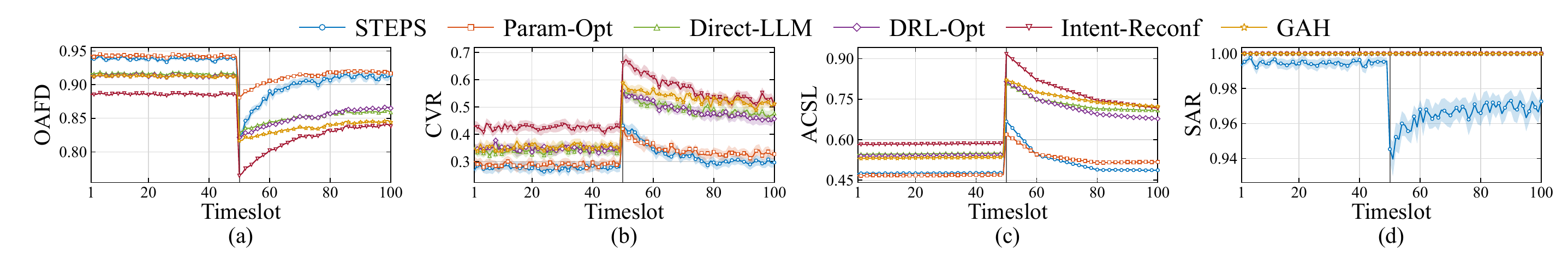}
	\caption{{\color{black}Scheduling performance under non-stationary drift: (a) per-slot OAFD, (b) per-slot CVR, (c) per-slot ACSL, and (d) per-slot SAR.}}
	\label{fig:drift_adaptation}
\end{figure*}

{\color{black}
	As shown in Fig.~\ref{fig:drift_adaptation}, the joint drift event occurs at $t_{\mathsf d}=50$, as indicated by the vertical dashed line. Fig.~\ref{fig:drift_adaptation}(a) shows that STEPS achieves OAFD close to Param-Opt and consistently higher than the other practical benchmarks before the drift. After the drift, its OAFD decreases sharply before gradually recovering, subsequently outperforming Direct-LLM, DRL-Opt, Intent-Reconf, and GAH. This behavior is supported by uncertainty-aware semantic admission and contract-guided scheduling, which help STEPS adapt to changing service requirements and resource conditions.
	Figs.~\ref{fig:drift_adaptation}(b) and (c) show that the drift initially increases CVR and ACSL for all methods, followed by gradual reductions over time. STEPS consistently exhibits the lowest CVR and ACSL during the later post-drift stage. This advantage is supported by semantic risk screening, uncertainty-aware contract calibration, congestion-aware best-response scheduling, and the fulfillment-driven adaptation of scheduling-control parameters. Together, these mechanisms help STEPS adjust its scheduling decisions to the evolving semantic requirements and edge resource conditions.
	Fig.~\ref{fig:drift_adaptation}(d) shows that the SAR of STEPS drops sharply at the drift event and then gradually recovers, whereas the benchmarks without semantic admission maintain near-complete admission throughout. The temporary reduction reflects the filtering of uncertain requests under the changed conditions, while the subsequent recovery is consistent with the admission-pressure feedback mechanism. Overall, STEPS exhibits favorable post-drift recovery in contract fulfillment and service loss, with a moderate tradeoff in admission coverage.
}

\subsubsection{Ablation Study}

{\color{black}
	To evaluate the contributions of uncertainty-aware contract handling and fulfillment-driven adaptation, we compare STEPS with two ablation variants, NoCal and NoEvo. NoCal disables semantic admission and uncertainty-aware contract calibration, whereas NoEvo retains these mechanisms but fixes the adaptive-control parameters and resource prices at their initial values. As shown in Table~\ref{tab:ablation}, STEPS achieves the highest AFD and OAFD, together with the lowest CVR and ACSL. Compared with NoEvo, STEPS further improves contract fulfillment and reduces service loss while maintaining comparable admission coverage. Although NoCal admits all requests, it exhibits lower fulfillment quality and higher service loss. These results demonstrate the contributions of uncertainty-aware contract calibration and fulfillment-driven adaptation to contract-guided scheduling.
}

{\color{black}
	\begin{table}[!t]
		\centering\color{black}
		\caption{Component ablation of STEPS.}
		\label{tab:ablation}
		\footnotesize
		\setlength{\tabcolsep}{2.2pt}
		\renewcommand{\arraystretch}{1.15}
		
		\begin{tabular*}{\columnwidth}{@{\extracolsep{\fill}}lrrrrr@{}}
			\toprule
			\textbf{Method} & \textbf{AFD}$\uparrow$ & \textbf{OAFD}$\uparrow$ & \textbf{CVR}$\downarrow$ & \textbf{ACSL}$\downarrow$ & \textbf{SAR}$\uparrow$ \\
			\midrule
			NoCal & 0.9226 & 0.9226 & 0.3346 & 0.5265 & \textbf{1.0000} \\
			NoEvo & 0.9407 & 0.9359 & 0.2945 & 0.4858 & 0.9949 \\
			\textbf{STEPS} & \textbf{0.9438} & \textbf{0.9388} & \textbf{0.2794} & \textbf{0.4758} & 0.9946 \\
			\bottomrule
		\end{tabular*}
		
		\par\smallskip
		\noindent\parbox{\columnwidth}{\footnotesize\textit{\textbf{Note:} Results are averaged over $30$ paired trials under the S2 configuration with $r_{\mathsf{unc}}=0.20$.}}
	\end{table}
}

\subsubsection{{\color{black}Post-Drift Feedback Ablation}}

{\color{black}
	We further evaluate the contribution of fulfillment-driven adaptation under non-stationary conditions by comparing STEPS with Frozen-After-Drift. Both methods share the same pre-drift history and contract-guided scheduling mechanism. After the drift event at $t_{\mathsf d}=50$, Frozen-After-Drift fixes the semantic admission threshold, contract conservativeness factor, coordination gain, and resource prices at their pre-drift values, whereas STEPS continues updating these parameters according to execution feedback. As shown in Table~\ref{tab:post-drift-feedback}, STEPS achieves higher AFD, OAFD, and SAR, while reducing CVR and ACSL. In particular, continued feedback updates improve OAFD by $0.0536$ and reduce ACSL by $0.1197$ compared with Frozen-After-Drift. These improvements demonstrate the effectiveness of fulfillment-driven adaptation in maintaining contract fulfillment and service coverage under non-stationary conditions.
}

{\color{black}
	\begin{table}[!t]
		\centering\color{black}
		\caption{Post-drift feedback ablation.}
		\label{tab:post-drift-feedback}
		\footnotesize
		\setlength{\tabcolsep}{2.2pt}
		\renewcommand{\arraystretch}{1.15}
		
		\begin{tabular*}{\columnwidth}{@{\extracolsep{\fill}}lrrrrr@{}}
			\toprule
			\textbf{Method} & \textbf{AFD}$\uparrow$ & \textbf{OAFD}$\uparrow$ & \textbf{CVR}$\downarrow$ & \textbf{ACSL}$\downarrow$ & \textbf{SAR}$\uparrow$ \\
			\midrule
			\textbf{STEPS} & \textbf{0.9285} & \textbf{0.8962} & \textbf{0.3256} & \textbf{0.5237} & \textbf{0.9652} \\
			Frozen & 0.8883 & 0.8426 & 0.4223 & 0.6434 & 0.9485 \\
			\bottomrule
		\end{tabular*}
		
		\par\smallskip
		\noindent\parbox{\columnwidth}{\footnotesize\textit{\textbf{Note:} Frozen denotes Frozen-After-Drift, which retains the pre-drift control parameters and resource prices after the drift event. Results are averaged over $30$ paired trials during timeslots $50$--$100$ under the S2 configuration. Paired statistical significance is assessed using the Wilcoxon signed-rank test with Holm correction.}}
	\end{table}
}

\subsection{{\color{black}Melbourne-based Evaluation}}
\label{sec:real_world_experiments}

{\color{black}
	We further evaluate STEPS in a Melbourne-based hybrid setting using geographic references from the Melbourne central business district (CBD)~\cite{EUAdata}. Specifically, four public-site coordinates are used to anchor ES locations, while the locations of $100$ UDs, natural-language requests, task workloads, wireless conditions, and T4E scores are synthetically generated. This hybrid setting incorporates geographic ES placement into the evaluation while retaining controlled service and network conditions. As shown in Table~\ref{tab:melbourne}, STEPS achieves the highest AFD and the lowest CVR among the six methods. It also maintains higher OAFD and lower ACSL than Direct-LLM, DRL-Opt, Intent-Reconf, and GAH. Param-Opt achieves slightly higher OAFD and lower ACSL by directly accessing numerical reference requirements. Overall, these results demonstrate that STEPS maintains competitive contract fulfillment and service loss under Melbourne-based edge-site placement. {\color{black}Additional semantic-interface validation and supplementary scheduling and adaptation results are provided in Appx.~\ref{app:semantic_validation} and Appx.~\ref{app:scheduling_results}, respectively.}
}

{\color{black}
	\begin{table}[!t]
		\centering\color{black}
		\caption{Melbourne-based evaluation.}
		\label{tab:melbourne}
		\footnotesize
		\setlength{\tabcolsep}{2.2pt}
		\renewcommand{\arraystretch}{1.15}
		
		\begin{tabular*}{\columnwidth}{@{\extracolsep{\fill}}lrrrrr@{}}
			\toprule
			\textbf{Method} & \textbf{AFD}$\uparrow$ & \textbf{OAFD}$\uparrow$ & \textbf{CVR}$\downarrow$ & \textbf{ACSL}$\downarrow$ & \textbf{SAR}$\uparrow$ \\
			\midrule
			Param-Opt & 0.9363 & \textbf{0.9363} & 0.3142 & \textbf{0.4838} & \textbf{1.0000} \\
			Direct-LLM & 0.9116 & 0.9116 & 0.3576 & 0.5583 & \textbf{1.0000} \\
			DRL-Opt & 0.9058 & 0.9058 & 0.3740 & 0.5560 & \textbf{1.0000} \\
			Intent-Reconf & 0.8779 & 0.8779 & 0.4542 & 0.6059 & \textbf{1.0000} \\
			GAH & 0.9069 & 0.9069 & 0.3740 & 0.5506 & \textbf{1.0000} \\
			\textbf{STEPS} & \textbf{0.9387} & 0.9323 & \textbf{0.2927} & 0.4880 & 0.9932 \\
			\bottomrule
		\end{tabular*}
		
		\par\smallskip
		\noindent\parbox{\columnwidth}{\footnotesize\textit{\textbf{Note:} Results are averaged over $30$ paired trials with $100$ simulated UDs and four geographically anchored ES locations.}}
	\end{table}
}

	\section{Conclusion}
	\label{sec:conclusion}
	
	{\color{black}We proposed STEPS, a semantic contract-guided edge scheduling framework integrating natural-language intent interpretation, potential-game-based scheduling, and feedback-driven adaptation, with theoretical guarantees on equilibrium existence, convergence, and parameter boundedness. Synthetic and Melbourne-based evaluations demonstrated competitive performance, characterized the fulfillment–admission tradeoff, and validated the benefits of continued adaptation under distribution shifts.} {\color{black}Future work will advance STEPS toward trustworthy, real-world semantic orchestration by grounding scheduling contracts in authentic user requests with independently validated service requirements. We further envision investigating end-to-end intent-to-execution fidelity, uncertainty-aware online adaptation, and robustness under dynamic workloads and heterogeneous edge environments, ultimately bridging the gap between high-level user intent and reliable, scalable edge service delivery.}

	\begin{spacing}{1}
		%\footnotesize
		\bibliographystyle{ieeetr}
		\bibliography{reference}

@article{mao2024green,
  title={Green edge {AI}: A contemporary survey},
  author={Mao, Yuyi and Yu, Xianghao and Huang, Kaibin and Zhang, Ying-Jun Angela and Zhang, Jun},
  journal={Proc. IEEE},
  volume={112},
  number={7},
  pages={880--911},
  year={2024},
  publisher={IEEE}
}

@article{he2025task,
  author={He, Yinghui and Li, Xin and Luo, Jun},
  journal={IEEE Trans. Mobile Comput.}, 
  title={Task-Oriented Integrated Sensing and Semantic Communications for Multi-Device Video Analytics}, 
  year={2026},
  volume={25},
  number={5},
  pages={7323-7337},
  publisher={IEEE}
}

@article{jia2025comprehensive,
  author={Jia, Ninghui and Qu, Zhihao and Ye, Baoliu and Wang, Yanyan and Hu, Shihong and Guo, Song},
  journal={IEEE Commun. Surveys Tut.}, 
  title={A Comprehensive Survey on Communication-Efficient Federated Learning in Mobile Edge Environments}, 
  year={2025},
  volume={27},
  number={6},
  pages={3710-3741},
  publisher={IEEE}
}

@article{qi2025future,
  title={Future Resource Bank for {ISAC}: Achieving Fast and Stable Win-Win Matching for Both Individuals and Coalitions},
   author={Qi, Houyi and Liwang, Minghui and Hosseinalipour, Seyyedali and Fu, Liqun and Zou, Sai and Ni, Wei},
  journal={IEEE J. Sel. Areas Commun.},  
  year={2026},
  volume={44},
  number={},
  pages={513-530},
  publisher={IEEE}
}

@article{fan2025latency,
  title={Latency-aware joint task offloading and energy control for cooperative mobile edge computing},
  author={Fan, Weibei and Xiao, Fu and Pan, Yao and Chen, Xiaobai and Han, Lei and Yu, Shui},
  journal={IEEE Trans. Serv. Comput.},
    year={2025},
  volume={18},
  number={3},
  pages={1515-1528},
}

@article{wang2025survey,
  title={A survey on intent-driven end-to-end {6G} mobile communication system},
  author={Wang, Yao and Yang, Chungang and Li, Tong and Ouyang, Ying and Mi, Xinru and Song, Yanbo},
  journal={IEEE Commun. Surveys Tut.},
    year={2026},
  volume={28},
  number={},
  pages={882-915},
  publisher={IEEE}
}

@inproceedings{mekrache2024llm,
  title={{LLM}-enabled intent-driven service configuration for next generation networks},
  author={Mekrache, Abdelkader and Ksentini, Adlen},
  booktitle={2024 IEEE Int. Conf. Netw. Softwarization (NetSoft)},
  pages={253--257},
  year={2024},
  organization={IEEE}
}

@article{qin2025generative,
  author={Qin, Xiaoqi and Sun, Mengying and Dai, Jincheng and Ma, Peixuan and Cao, Yuecheng and Zhang, Jingjing and Wang, Jiacheng and Xu, Xiaodong and Zhang, Ping and Niyato, Dusit},
  journal={IEEE Trans. Cogn. Commun. Netw.}, 
  title={Generative {AI} Meets Wireless Networking: An Interactive Paradigm for Intent-Driven Communications}, 
  year={2025},
  volume={11},
  number={4},
  pages={2056-2077},
  organization={IEEE}
}

@article{sun2026igaa,
  title={{IGAA}: Intent-Driven General Agentic {AI} for Edge Services Scheduling using Generative Meta Learning},
  author={Sun, Yan and Liu, Yinqiu and Guo, Shaoyong and Zhang, Ruichen and Qi, Feng and Qiu, Xuesong and Gong, Weifeng and Niyato, Dusit and Wu, Qihui},
  journal={arXiv:2601.13702},
  year={2026}
}

@article{kalntis2024adaptive,
  title={Adaptive resource allocation for virtualized base stations in {O-RAN} with online learning},
  author={Kalntis, Michail and Iosifidis, George and Kuipers, Fernando A},
  journal={IEEE Trans. Commun.},
  volume={73},
  number={3},
  pages={1787--1800},
  year={2025},
  publisher={IEEE}
}

@inproceedings{jacobs2021hey,
  title={Hey, {Lumi}! using natural language for $\{$intent-based$\}$ network management},
  author={Jacobs, Arthur S and Pfitscher, Ricardo J and Ribeiro, Rafael H and Ferreira, Ronaldo A and Granville, Lisandro Z and Willinger, Walter and Rao, Sanjay G},
  booktitle={2021 Usenix Annu. Tech. Conf. (Usenix atc 21)},
  pages={625--639},
  year={2021}
}

@article{mekrache2024intent,
  title={Intent-based management of next-generation networks: An {LLM}-centric approach},
  author={Mekrache, Abdelkader and Ksentini, Adlen and Verikoukis, Christos},
  journal={IEEE Netw.},
  volume={38},
  number={5},
  pages={29--36},
  year={2024},
  publisher={IEEE}
}

@inproceedings{wang2025intent,
  title={Intent-driven network management with multi-agent {LLMs}: The confucius framework},
  author={Wang, Zhaodong and Lin, Samuel and Yan, Guanqing and Ghorbani, Soudeh and Yu, Minlan and Zhou, Jiawei and Hu, Nathan and Baruah, Lopa and Peters, Sam and Kamath, Srikanth and others},
  booktitle={Proc. ACM SIGCOMM 2025 Conf.},
  pages={347--362},
  year={2025}
}

@inproceedings{akbari2025intentcontinuum,
  title={Intentcontinuum: Using {LLM}s to support intent-based computing across the compute continuum},
  author={Akbari, Negin and Grundy, John and Cheema, Aamir and Toosi, Adel N},
  booktitle={IEEE Int. Conf. Web Serv.},
  pages={573--583},
  year={2025},
}

@article{xu2024unleashing,
  title={Unleashing the power of edge-cloud generative {AI} in mobile networks: A survey of {AIGC} services},
  author={Xu, Minrui and Du, Hongyang and Niyato, Dusit and Kang, Jiawen and Xiong, Zehui and Mao, Shiwen and Han, Zhu and Jamalipour, Abbas and Kim, Dong In and Shen, Xuemin and others},
  journal={IEEE Commun. Surveys Tut.},
  volume={26},
  number={2},
  pages={1127--1170},
  year={2024},
  publisher={IEEE}
}

@article{xiao2024adaptive,
  title={Adaptive compression offloading and resource allocation for edge vision computing},
  author={Xiao, Wenjing and Hao, Yixue and Liang, Junbin and Hu, Long and Alqahtani, Salman A and Chen, Min},
  journal={IEEE Trans. Cogn. Commun. Netw.},
  volume={10},
  number={6},
  pages={2357--2369},
  year={2024},
  publisher={IEEE}
}

@article{li2024optimal,
  title={Optimal {AI} model splitting and resource allocation for device-edge co-inference in multi-user wireless sensing systems},
  author={Li, Xian and Bi, Suzhi},
  journal={IEEE Trans. Wireless Commun.},
  volume={23},
  number={9},
  pages={11094--11108},
  year={2024},
  publisher={IEEE}
}

@article{zhang2025beyond,
  title={Beyond the cloud: Edge inference for generative large language models in wireless networks},
  author={Zhang, Xinyuan and Nie, Jiangtian and Huang, Yudong and Xie, Gaochang and Xiong, Zehui and Liu, Jiang and Niyato, Dusit and Shen, Xuemin},
  journal={IEEE Trans. Wireless Commun.},
  volume={24},
  number={1},
  pages={643--658},
  year={2025},
  publisher={IEEE}
}

@article{qu2025mobile,
  title={Mobile edge intelligence for large language models: A contemporary survey},
  author={Qu, Guanqiao and Chen, Qiyuan and Wei, Wei and Lin, Zheng and Chen, Xianhao and Huang, Kaibin},
  journal={IEEE Commun. Surveys Tut.},
  volume={27},
  number={6},
  pages={3820--3860},
  year={2025},
  publisher={IEEE}
}

@article{ganguly2024online,
  author={Ganguly, Bhargav and Aggarwal, Vaneet},
  journal={IEEE/ACM Trans. Netw.}, 
  title={Online Federated Learning via Non-Stationary Detection and Adaptation Amidst Concept Drift}, 
  year={2024},
  volume={32},
  number={1},
  pages={643-653},
  publisher={IEEE}
}

@article{gudepu2024drift,
  title={The drift handling framework for open radio access networks: An experimental evaluation},
  author={Gudepu, Venkateswarlu and Chintapalli, Venkatarami Reddy and Castoldi, Piero and Valcarenghi, Luca and Tamma, Bheemarjuna Reddy and Kondepu, Koteswararao},
  journal={Comput. Netw.},
  volume={243},
  pages={110290},
  year={2024},
  publisher={Elsevier}
}

@article{uzlaner2025asynchronous,
  title={Asynchronous online adaptation via modular drift detection for deep receivers},
  author={Uzlaner, Nicole and Raviv, Tomer and Shlezinger, Nir and Todros, Koby},
  journal={IEEE Trans. Wireless Commun.},
  volume={24},
  number={5},
  pages={4454--4468},
  year={2025},
  publisher={IEEE}
}

@article{ameur2025dual,
  title={Dual self-attention is what you need for model drift detection in {6G} networks},
  author={Ameur, Mazene and Brik, Bouziane and Ksentini, Adlen},
  journal={IEEE Trans. Mach. Learn. Commun. Netw.},
  year={2025},
  volume={3},
  number={},
  pages={690-709},
  publisher={IEEE}
}

@ARTICLE{EUAdata,
	author={He, Qiang and Cui, Guangming and Zhang, Xuyun and Chen, Feifei and Deng, Shuiguang and Jin, Hai and Li, Yanhui and Yang, Yun},
	journal={IEEE Trans. Parallel Distrib. Syst.}, 
	title={A Game-Theoretical Approach for User Allocation in Edge Computing Environment}, 
	year={2020},
	volume={31},
	number={3},
	pages={515-529},
	publisher={IEEE}
	}

@ARTICLE{rosenthal1973class,
author = {Rosenthal, Robert W.},
title = {A class of games possessing pure-strategy Nash equilibria},
year = {1973},
issue_date = {Dec 1973},
publisher = {Physica-Verlag GmbH},
address = {DEU},
volume = {2},
number = {1},
issn = {0020-7276},
url = {https://doi.org/10.1007/BF01737559},
doi = {10.1007/BF01737559},
journal = {Int. J. Game Theory},
month = dec,
pages = {65–67},
numpages = {3}}

@ARTICLE{kelly1998rate,
author = {F P Kelly and A K Maulloo and D K H Tan},
title = {Rate control for communication networks: shadow prices, proportional fairness and stability},
journal = {J. Oper. Res. Soc.},
volume = {49},
number = {3},
pages = {237--252},
year = {1998},
publisher = {Taylor \& Francis},
doi = {10.1057/palgrave.jors.2600523}}

@ARTICLE{wang2024cooperative,
  author={Wang, Yuao and Fang, Jingjing and Cheng, Yao and She, Hao and Guo, Yongan and Zheng, Gan},
  journal={IEEE J. Sel. Topics Signal Process.}, 
  title={Cooperative End-Edge-Cloud Computing and Resource Allocation for Digital Twin Enabled {6G} Industrial {IoT}}, 
  year={2024},
  volume={18},
  number={1},
  pages={124-137},
  publisher={IEEE}
}

@ARTICLE{zhang2025communication,
  author={Zhang, Kai and He, Hengtao and Song, Shenghui and Zhang, Jun and Letaief, Khaled B.},
  journal={IEEE J. Sel. Topics Signal Process.}, 
  title={Communication-Efficient Distributed On-Device {LLM} Inference Over Wireless Networks}, 
  year={2025},
  volume={19},
  number={7},
  pages={1301-1317},
}

@ARTICLE{zhang2025split,
  author={Zhang, Songge and Cheng, Guoliang and Wu, Wen and Huang, Xinyu and Song, Lingyang and Shen, Xuemin},
  journal={IEEE J. Sel. Topics Signal Process.}, 
  title={Split Fine-Tuning for Large Language Models in Wireless Networks}, 
  year={2025},
  volume={19},
  number={7},
  pages={1376-1391},
    publisher={IEEE}
}

@ARTICLE{ji2024computational,
  author={Ji, Zelin and Qin, Zhijin},
  journal={IEEE J. Sel. Topics Signal Process.}, 
  title={Computational Offloading in Semantic-Aware Cloud-Edge-End Collaborative Networks}, 
  year={2024},
  volume={18},
  number={7},
  pages={1235-1248},
  publisher={IEEE}
}

@ARTICLE{Abdisarabshali2026Hierarchical,
  author={Abdisarabshali, Payam and Nadimi, Fardis and Borazjani, Kasra and Khosravan, Naji and Liwang, Minghui and Ni, Wei and Niyato, Dusit and Langberg, Michael and Hosseinalipour, Seyyedali},
  journal={IEEE Commun. Mag.}, 
  title={Hierarchical Federated Foundation Models over Wireless Networks for Multi-Modal Multi-Task Intelligence: Integration of Edge Learning with {D2D/P2P}-Enabled Fog Learning Architectures}, 
  year={2026},
  volume={64},
  number={4},
  pages={66-72},
  publisher={IEEE}
}

@inproceedings{larson-etal-2019-evaluation,
	author    = {Stefan Larson and Anish Mahendran and Joseph J. Peper and Christopher Clarke and Andrew Lee and Parker Hill and Jonathan K. Kummerfeld and Kevin Leach and Michael A. Laurenzano and Lingjia Tang and Jason Mars},
	title     = {An Evaluation Dataset for Intent Classification and Out-of-Scope Prediction},
	booktitle = {Proc. EMNLP-IJCNLP},
	pages     = {1311--1316},
	year      = {2019},
	doi       = {10.18653/v1/D19-1131},
}

@inproceedings{hartvigsen-etal-2022-toxigen,
	author    = {Thomas Hartvigsen and Saadia Gabriel and Hamid Palangi and Maarten Sap and Dipankar Ray and Ece Kamar},
	title     = {{ToxiGen}: A Large-Scale Machine-Generated Dataset for Adversarial and Implicit Hate Speech Detection},
	booktitle = {Proc. 60th Annu. Meeting Assoc. Comput. Linguist. (ACL)},
	pages     = {3309--3326},
	year      = {2022},
	doi       = {10.18653/v1/2022.acl-long.234},
}

@inproceedings{lin-etal-2024-generate,
	author    = {I-Fan Lin and Faegheh Hasibi and Suzan Verberne},
	title     = {Generate then Refine: Data Augmentation for Zero-shot Intent Detection},
	booktitle = {Findings Assoc. Comput. Linguist.: EMNLP 2024},
	pages     = {13138--13146},
	year      = {2024},
	doi       = {10.18653/v1/2024.findings-emnlp.768},
}

@inproceedings{frank-afli-2026-pteb,
	author    = {Manuel Frank and Haithem Afli},
	title     = {{PTEB}: Towards Robust Text Embedding Evaluation via Stochastic Paraphrasing at Evaluation Time with {LLM}s},
	booktitle = {Proc. 19th Conf. Eur. Chapter Assoc. Comput. Linguist. (EACL)},
	pages     = {2832--2851},
	year      = {2026},
	doi       = {10.18653/v1/2026.eacl-long.130}
}
	\end{spacing}
	
	\newpage
	\clearpage
	\appendices

	{\color{black}
		\section{Semantic Interface Validation}
		\label{app:semantic_validation}
		
		\setcounter{figure}{0}
		\setcounter{table}{0}
		\renewcommand{\thefigure}{\Alph{section}\arabic{figure}}
		\renewcommand{\thetable}{\Alph{section}\arabic{table}}
		\renewcommand{\theHfigure}{app.\Alph{section}.\arabic{figure}}
		\renewcommand{\theHtable}{app.\Alph{section}.\arabic{table}}

\subsection{Dataset Construction and Semantic Mapping}

We evaluate the semantic interface of STEPS using natural-language requests beyond the cached scheduling templates employed in the main experiments. The evaluation includes $120$ calibration requests and $180$ held-out clear requests covering unseen paraphrases, free-form expressions, and stringent yet unambiguous service requirements. We further consider ambiguous, conflicting, and out-of-domain (OOD) requests, with OOD samples drawn from CLINC150~\cite{larson-etal-2019-evaluation}. The four semantic dimensions correspond to delay, energy consumption, monetary cost, and trustworthiness-for-execution (T4E).

Following related studies on LLM-assisted data construction~\cite{hartvigsen-etal-2022-toxigen,lin-etal-2024-generate}, the controlled requests are generated with LLM assistance and subsequently reviewed. Complete four-dimensional reference labels are assigned to the calibration and clear-test requests, while ambiguous, conflicting, and OOD requests are evaluated using their respective category labels. Table~\ref{tab:dataset-contract}(a) summarizes the request collections.

The LLM-assisted semantic parser extracts semantic service levels and raw confidence scores from these requests. The extracted levels are converted into preference weights and original fulfillment bounds according to~\eqref{eq:weight_mapping} and~\eqref{eq:bound_mapping}. Specifically, the fixed reference policy $\pi^{\mathsf{ref}}$ selects the nearest edge server (ES) with predefined computing and bandwidth allocations to estimate the reference service indicators $\bar{D}_{i}^{(t)}$, $\bar{E}_{i}^{(t)}$, and $\bar{M}_{i}^{(t)}$. Table~\ref{tab:dataset-contract}(b) summarizes the corresponding semantic-level mappings. The context-inconsistency risk in~\eqref{eq:context_risk} further evaluates the compatibility between the original fulfillment requirements and individually feasible execution actions.

\begin{table}[!b]
	\centering\color{black}
	\caption{Dataset construction and semantic mappings.}
	\label{tab:dataset-contract}
	\footnotesize
	\setlength{\tabcolsep}{3pt}
	\renewcommand{\arraystretch}{1.15}
	
	\begin{tabular*}{\columnwidth}{@{\extracolsep{\fill}}lccc@{}}
		\toprule
		\multicolumn{4}{@{}l}{\textbf{(a) Natural-language request collections}} \\
		\midrule
		\textbf{Category} & \textbf{Requests} & \textbf{Families} & \textbf{Reference} \\
		\midrule
		Calibration        & 120 & 20 & Yes \\
		Unseen paraphrases & 60  & 10 & Yes \\
		Free-form          & 60  & 10 & Yes \\
		Stringent clear    & 60  & 10 & Yes \\
		Ambiguous          & 60  & 12 & No  \\
		Conflicting        & 60  & 12 & No  \\
		General OOD        & 60  & 60 & No  \\
		\bottomrule
	\end{tabular*}
	
	\vspace{2mm}
	
	\begin{tabular*}{\columnwidth}{@{\extracolsep{\fill}}lccc@{}}
		\toprule
		\multicolumn{4}{@{}l}{\textbf{(b) Semantic-level mappings}} \\
		\midrule
		\textbf{Quantity} & \textbf{Low} & \textbf{Medium} & \textbf{High} \\
		\midrule
		$\hat{\ell}_{i}^{\mathsf{k},(t)}$ & $0$ & $0.5$ & $1$ \\
		$D_{i}^{\mathsf{max},(t)}$ & $1.2\bar{D}_{i}^{(t)}$ & $\bar{D}_{i}^{(t)}$ & $0.8\bar{D}_{i}^{(t)}$ \\
		$E_{i}^{\mathsf{max},(t)}$ & $1.2\bar{E}_{i}^{(t)}$ & $\bar{E}_{i}^{(t)}$ & $0.8\bar{E}_{i}^{(t)}$ \\
		$M_{i}^{\mathsf{max},(t)}$ & $1.2\bar{M}_{i}^{(t)}$ & $\bar{M}_{i}^{(t)}$ & $0.8\bar{M}_{i}^{(t)}$ \\
		$S_{i}^{\mathsf{min},(t)}$ & $0.5$ & $0.7$ & $0.9$ \\
		\bottomrule
	\end{tabular*}
	
	\par\smallskip
	\noindent\parbox{\columnwidth}{\footnotesize\textit{\textbf{Note:} The reference column in panel (a) indicates the availability of complete four-dimensional semantic labels. In panel (b), $\mathsf{k}\in\{\mathsf{D},\mathsf{E},\mathsf{M},\mathsf{S}\}$. The nearest-ES reference uses computing and bandwidth shares of $0.08$, with ties resolved by ES index. The preference coefficient is $\theta_w=2$, and $\omega_{\mathsf{k}}=0.25$ for all four dimensions. The T4E mapping corresponds to $S_0=0.5$ and $\delta_{\mathsf{S}}=0.4$.}}
\end{table}

\subsection{Semantic Extraction and Uncertainty Evaluation}

Fig.~\ref{fig:parser_validation}(a) reports the dimension-wise macro-averaged F1 scores (Macro-F1) of Qwen3.5-4B across the three clear-request categories. The parser achieves consistently high extraction performance for monetary cost and T4E, while free-form requests exhibit lower accuracy in delay and energy sensitivity. Among the $180$ held-out clear requests, $132$ are correctly interpreted across all four dimensions, yielding a joint accuracy of $73.33\%$.

Fig.~\ref{fig:parser_validation}(b) evaluates the reliability of parsing confidence against empirical joint correctness. The raw confidence scores tend to overestimate prediction correctness, whereas the monotonic logit-affine calibration fitted on the independent calibration set brings the confidence estimates closer to the observed accuracy. This calibration improves confidence reliability without changing the predicted semantic levels.

Fig.~\ref{fig:parser_validation}(c) compares the area under the receiver operating characteristic curve (AUROC) for different uncertainty-detection tasks. The combined semantic uncertainty measure achieves relatively high AUROC for ambiguous and conflicting requests, while its discrimination of ordinary parsing errors and general OOD requests is weaker. Five-call self-consistency provides stronger detection of parsing errors, although its performance on general OOD requests remains limited. Fig.~\ref{fig:parser_validation}(d) further shows that parser-derived risk varies relatively little across request categories, whereas context-inconsistency risk increases for ambiguous and conflicting requests. By incorporating both calibrated parsing confidence and context-inconsistency risk, STEPS captures semantic uncertainty associated with interpretation reliability and the feasibility of inferred service requirements.

\begin{figure*}[!t]
	\centering\color{black}
	\includegraphics[width=0.98\textwidth]{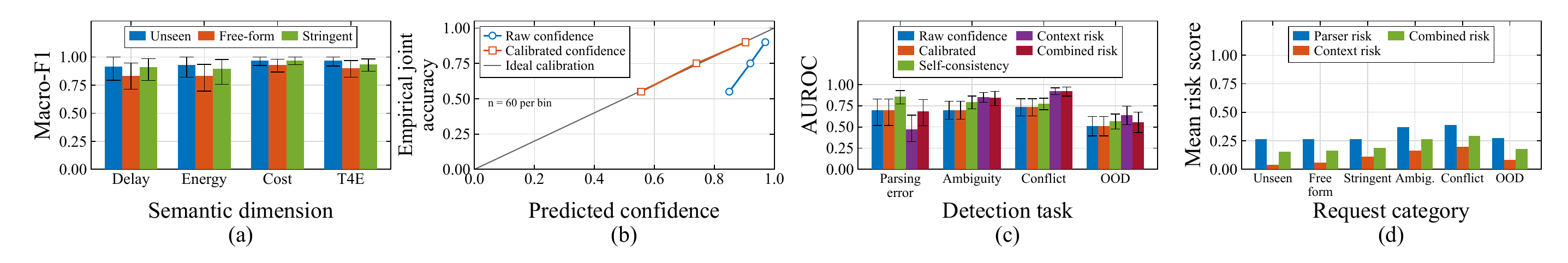}
	\caption{Semantic extraction and uncertainty evaluation: (a) dimension-wise Macro-F1, (b) confidence reliability, (c) uncertainty-detection AUROC, and (d) semantic risk components across request categories.}
	\label{fig:parser_validation}
\end{figure*}

\subsection{Downstream Impact of Semantic Parsing Errors}

We next examine the impact of semantic parsing errors on contract-guided scheduling using the $180$ held-out clear requests and $30$ paired network realizations under the S2 configuration, comprising $100$ user devices and four edge servers. We compare two variants: \textit{Oracle} constructs semantic contracts from reference labels, whereas \textit{Predicted} uses the corresponding LLM parsing outputs. Both variants operate without semantic admission and uncertainty-aware contract calibration. We evaluate the resulting changes in execution-node (EN) selection, final resource configurations, and contract fulfillment using AFD, CVR, and ACSL as defined in Sec.~\ref{sec:metrics}.

Among the $180$ held-out clear requests, $48$ contain at least one incorrectly parsed semantic level. Tables~\ref{tab:semantic-robustness}(b) and (c) examine representative errors by comparing the scheduling decisions and fulfillment outcomes obtained from reference and predicted semantic labels. These errors affect both execution-node (EN) selection and final resource configurations. In particular, interpreting high delay sensitivity as low changes the EN selection in $33.70\%$ of the paired cases and the final resource configuration in $67.78\%$. This error reduces AFD by $0.0816$ while increasing CVR and ACSL by $0.3778$ and $0.1331$, respectively. These results show how incorrect semantic requirements can alter contract-guided scheduling decisions and degrade service fulfillment.

Table~\ref{tab:semantic-robustness}(a) further compares the agreement between the primary parser and alternative model or prompt configurations. The alternative prompts using Qwen3.5-4B achieve higher joint agreement with the primary parser than Qwen3.5-2B. Among the error types in Tables~\ref{tab:semantic-robustness}(b) and (c), the high-to-low delay-sensitivity error produces the largest changes in final resource configurations and the greatest reduction in AFD. These findings highlight the importance of accurate semantic interpretation for subsequent execution decisions and contract fulfillment.

\begin{table}[!t]
	\centering\color{black}
	\caption{Semantic robustness and error propagation.}
	\label{tab:semantic-robustness}
	\footnotesize
	\setlength{\tabcolsep}{2pt}
	\renewcommand{\arraystretch}{1.15}
	
	\begin{tabular*}{\columnwidth}{@{\extracolsep{\fill}}lrrrrr@{}}
		\toprule
		\multicolumn{6}{@{}l}{\textbf{(a) Agreement with the primary parser (\%)}} \\
		\midrule
		\textbf{Variant} & $\mathsf{D}$ & $\mathsf{E}$ & $\mathsf{M}$ & \textbf{T4E} & \textbf{Joint} \\
		\midrule
		Qwen3.5-2B    & 90.0 & 85.0 & 91.7 & 88.3 & 60.0 \\
		4B, prompt 1  & 98.3 & 91.7 & 93.3 & 96.7 & 81.7 \\
		4B, prompt 2  & 91.7 & 90.0 & 90.0 & 90.0 & 66.7 \\
		\bottomrule
	\end{tabular*}
	
	\vspace{2mm}
	
	\begin{tabular*}{\columnwidth}{@{\extracolsep{\fill}}lrrr@{}}
		\toprule
		\multicolumn{4}{@{}l}{\textbf{(b) Scheduling changes under parsing errors}} \\
		\midrule
		\textbf{Error} & \textbf{Count} & \textbf{EN (\%)} & \textbf{Config. (\%)} \\
		\midrule
		Delay: Medium$\rightarrow$Low & 8 & 17.50 & 34.58 \\
		Delay: High$\rightarrow$Low   & 9 & 33.70 & 67.78 \\
		Cost: Low$\rightarrow$High    & 8 & 12.50 & 35.42 \\
		\bottomrule
	\end{tabular*}
	
	\vspace{2mm}
	
	\begin{tabular*}{\columnwidth}{@{\extracolsep{\fill}}lrrr@{}}
		\toprule
		\multicolumn{4}{@{}l}{\textbf{(c) Fulfillment effects of parsing errors}} \\
		\midrule
		\textbf{Error} & $\Delta\mathsf{AFD}$ & $\Delta\mathsf{CVR}$ & $\Delta\mathsf{ACSL}$ \\
		\midrule
		Delay: Medium$\rightarrow$Low & -0.0368 & 0.1667 & 0.0529 \\
		Delay: High$\rightarrow$Low   & -0.0816 & 0.3778 & 0.1331 \\
		Cost: Low$\rightarrow$High    & -0.0473 & 0.2250 & 0.0751 \\
		\bottomrule
	\end{tabular*}
	
	\par\smallskip
	\noindent\parbox{\columnwidth}{\footnotesize\textit{\textbf{Note:} Panel (a) reports agreement with the primary parser for $60$ requests per variant. The 4B variants use Qwen3.5-4B with alternative prompts. Panels (b) and (c) report results over $30$ paired S2 network realizations. EN and Config. denote changes in execution-node selection and final resource configurations, respectively. The differences in panel (c) are calculated as Predicted minus Oracle. Error groups may overlap.}}
\end{table}
	}

	{\color{black}
		\section{Supplementary Scheduling and Adaptation Results}
		\label{app:scheduling_results}
		
		\setcounter{figure}{0}
		\setcounter{table}{0}
		\renewcommand{\thefigure}{\Alph{section}\arabic{figure}}
		\renewcommand{\thetable}{\Alph{section}\arabic{table}}
		\renewcommand{\theHfigure}{app.\Alph{section}.\arabic{figure}}
		\renewcommand{\theHtable}{app.\Alph{section}.\arabic{table}}

	\subsection{Complementary Performance and T4E Evaluation}
	
	Figs.~\ref{fig:supp_performance}(a) and (b) further compare contract violation rate (CVR) under increasing semantic uncertainty and average fulfillment degree (AFD) across different problem scales, respectively. STEPS maintains relatively low CVR as the uncertain-request ratio $r_{\mathsf{unc}}$ increases, while achieving the highest AFD across the four network scales. This performance is supported by uncertainty-aware semantic admission and contract-guided scheduling, which account for semantic reliability and multidimensional service requirements.
	
	Figs.~\ref{fig:supp_performance}(c) and (d) examine fulfillment performance under non-stationary drift. Following the joint drift at $t_{\mathsf d}=50$, AFD initially decreases and then gradually recovers, with STEPS achieving the highest fulfillment toward the end of the evaluation. To assess post-drift service loss, we define the post-drift running-average ACSL (PD-ACSL) as {\small$\mathsf{PD\text{-}ACSL}_m^{(t)}=\frac{1}{t-t_{\mathsf d}+1}\sum_{\tau=t_{\mathsf d}}^{t}\mathsf{ACSL}_m^{(\tau)}$}, for $t\ge t_{\mathsf d}$, where $\mathsf{ACSL}_m^{(\tau)}$ denotes the average service loss of served requests at timeslot $\tau$. As shown in Fig.~\ref{fig:supp_performance}(d), STEPS progressively reduces its post-drift average service loss and achieves the lowest PD-ACSL toward the end of the evaluation, reflecting the effectiveness of fulfillment-driven adaptation under changing service and resource conditions.
	
	\begin{figure*}[!t]
		\centering\color{black}
		\includegraphics[width=0.98\textwidth]{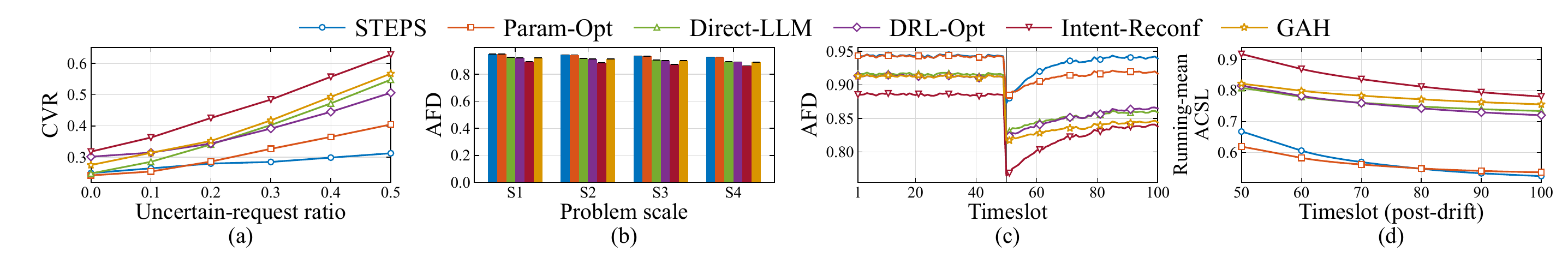}
		\caption{Complementary scheduling performance: (a) CVR under increasing semantic uncertainty, (b) AFD across different problem scales, (c) per-slot AFD under non-stationary drift, and (d) post-drift running-average ACSL.}
		\label{fig:supp_performance}
	\end{figure*}
	
	We further examine the contribution of the trustworthiness-for-execution (T4E) dimension by comparing STEPS with No-T4E in Table~\ref{tab:t4e_ablation}. No-T4E excludes T4E requirements from scheduling decisions, while both methods are evaluated against the same reference requirements. Under low T4E requirements, the two methods achieve similar AFD and T4E scores. As the requirements become more stringent, STEPS achieves higher fulfillment and T4E scores while substantially reducing T4E violations. In particular, under high T4E requirements, STEPS improves AFD from $0.8664$ to $0.9279$, increases the achieved T4E score from $0.7996$ to $0.9297$, and reduces the T4E violation rate from $0.3910$ to $0.0757$ compared with No-T4E. The corresponding SAR decreases from $1.0000$ to $0.9798$, reflecting the tradeoff between trustworthiness-related fulfillment and admission coverage.
	
	\begin{table}[!t]
		\centering\color{black}
		\caption{Performance under heterogeneous T4E requirements.}
		\label{tab:t4e_ablation}
		\footnotesize
		\setlength{\tabcolsep}{2.2pt}
		\renewcommand{\arraystretch}{1.15}
		
		\begin{tabular*}{\columnwidth}{@{\extracolsep{\fill}}llrrrr@{}}
			\toprule
			\textbf{Level} & \textbf{Method} & \textbf{AFD}$\uparrow$ & \textbf{SAR}$\uparrow$ & \textbf{T4E-V}$\downarrow$ & \textbf{T4E}$\uparrow$ \\
			\midrule
			Low & \textbf{STEPS} & \textbf{0.9538} & 0.9979 & \textbf{0.0148} & 0.7798 \\
			Low & No-T4E & 0.9517 & \textbf{1.0000} & 0.0177 & \textbf{0.7801} \\
			\midrule
			Medium & \textbf{STEPS} & \textbf{0.9440} & 0.9949 & \textbf{0.0349} & \textbf{0.8400} \\
			Medium & No-T4E & 0.9218 & \textbf{1.0000} & 0.1698 & 0.7801 \\
			\midrule
			High & \textbf{STEPS} & \textbf{0.9279} & 0.9798 & \textbf{0.0757} & \textbf{0.9297} \\
			High & No-T4E & 0.8664 & \textbf{1.0000} & 0.3910 & 0.7996 \\
			\bottomrule
		\end{tabular*}
		
		\par\smallskip
		\noindent\parbox{\columnwidth}{\footnotesize\textit{\textbf{Note:} Results are averaged over $30$ paired trials. T4E-V denotes the T4E violation rate, while T4E denotes the average achieved trustworthiness-for-execution score. The low, medium, and high levels correspond to separate requirement groups, with reference T4E bounds of $0.5$, $0.7$, and $0.9$, respectively. SAR measures the fraction of requests admitted for scheduling.}}
	\end{table}
	
	\subsection{Convergence and Resource Realization}
	
	We next examine the convergence of asynchronous best-response updates and the effect of capacity-feasible resource realization. As shown in Table~\ref{tab:table_c3}(a), the mean number of completed best-response rounds increases from $2.1937$ at S1 to $7.1740$ at S4. The maximum observed round count is $14$, and none of the evaluated instances reaches the prescribed limit $I_{\mathsf{max}}=20$. These results show that the best-response procedure terminates within the adopted iteration budget across the tested network scales.
	
	Table~\ref{tab:table_c3}(b) examines proportional resource realization. When the intended aggregate resource demand exceeds an ES's available capacity, the computing and bandwidth allocations are scaled according to~\eqref{eq:compute_realization_factor} and~\eqref{eq:bandwidth_realization_factor}. As the problem scale increases, the proportion of requests affected by scaling rises from $4.24\%$ to $15.20\%$, while the proportion of ES-timeslot instances requiring scaling increases from $12.53\%$ to $42.26\%$. Nevertheless, the average realization factors remain close to one, and the realized-minus-intended AFD difference ranges from $-0.0002$ to $-0.0011$. These results indicate that proportional resource realization maintains aggregate capacity feasibility with relatively small fulfillment losses across the tested problem scales.
	
	\begin{table}[!t]
		\centering\color{black}
		\caption{Convergence and resource realization.}
		\label{tab:table_c3}
		\footnotesize
		\setlength{\tabcolsep}{2pt}
		\renewcommand{\arraystretch}{1.15}
		
		\begin{tabular*}{\columnwidth}{@{\extracolsep{\fill}}lrrrr@{}}
			\toprule
			\multicolumn{5}{@{}l}{\textbf{(a) Best-response convergence}} \\
			\midrule
			\textbf{Scale} & \textbf{Mean} & \textbf{P95} & \textbf{Max} & \textbf{Cap (\%)} \\
			\midrule
			S1 & 2.1937 & 4 & 7 & 0 \\
			S2 & 3.8460 & 6 & 9 & 0 \\
			S3 & 5.5487 & 9 & 12 & 0 \\
			S4 & 7.1740 & 10 & 14 & 0 \\
			\bottomrule
		\end{tabular*}
		
		\vspace{2mm}
		
		\begin{tabular*}{\columnwidth}{@{\extracolsep{\fill}}lrrrrr@{}}
			\toprule
			\multicolumn{6}{@{}l}{\textbf{(b) Capacity-feasible resource realization}} \\
			\midrule
			\textbf{Scale} & \textbf{Req. (\%)} & \textbf{ES (\%)} & $\overline{\kappa}^{\mathsf f}$ & $\overline{\kappa}^{\mathsf b}$ & $\Delta\mathrm{AFD}$ \\
			\midrule
			S1 & 4.24 & 12.53 & 0.9975 & 0.9956 & -0.0002 \\
			S2 & 6.69 & 19.24 & 0.9952 & 0.9923 & -0.0004 \\
			S3 & 11.26 & 31.22 & 0.9906 & 0.9861 & -0.0007 \\
			S4 & 15.20 & 42.26 & 0.9853 & 0.9788 & -0.0011 \\
			\bottomrule
		\end{tabular*}
		
		\par\smallskip
		\noindent\parbox{\columnwidth}{\footnotesize\textit{\textbf{Note:} Panel (a) reports the number of completed best-response rounds. P95 and Max denote the $95$th percentile and maximum round count, respectively, while Cap is the percentage of instances reaching $I_{\mathsf{max}}=20$. In panel (b), Req. and ES denote the percentages of scaled requests and ES-timeslot instances requiring scaling. The average realization factors $\overline{\kappa}^{\mathsf f}$ and $\overline{\kappa}^{\mathsf b}$ include unscaled instances with a factor of one. $\Delta\mathrm{AFD}$ denotes realized AFD minus intended AFD.}}
	\end{table}

\subsection{Parameter Sensitivity and Feedback Adaptation}

\begin{figure*}[!t]
	\centering\color{black}
	\includegraphics[width=0.98\textwidth]{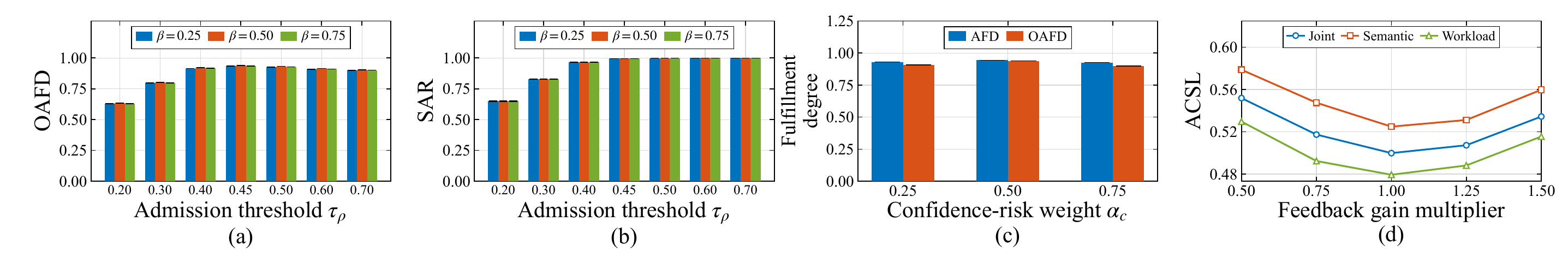}
	\caption{Sensitivity to semantic uncertainty and feedback parameters: (a) OAFD and (b) SAR under different admission thresholds and contract conservativeness factors, (c) AFD and OAFD versus the confidence-risk weight, and (d) ACSL under different feedback-gain multipliers and drift conditions.}
	\label{fig:parameter_sensitivity}
\end{figure*}

Figs.~\ref{fig:parameter_sensitivity}(a) and (b) examine the effects of the semantic admission threshold $\tau^{\rho,(t)}$ and contract conservativeness factor $\beta^{(t)}$ on OAFD and SAR. As the admission threshold increases, SAR gradually approaches one, while OAFD improves substantially at lower thresholds and reaches its highest level near the intermediate settings. Further increasing the threshold slightly reduces OAFD despite the higher admission coverage. In comparison, varying $\beta^{(t)}$ produces relatively small performance differences. These results indicate that semantic admission plays a key role in balancing contract fulfillment and service coverage.

Fig.~\ref{fig:parameter_sensitivity}(c) compares AFD and OAFD under different confidence-risk weights $\alpha_c$ in~\eqref{eq:semantic_uncertainty}. Both metrics exhibit relatively small variations, with their highest values occurring near the nominal configuration. This suggests that balancing calibrated parsing confidence and context-inconsistency risk helps maintain stable contract fulfillment. Fig.~\ref{fig:parameter_sensitivity}(d) further evaluates ACSL under different feedback-gain multipliers for joint, semantic, and workload drift. The nominal gain configuration achieves the lowest service loss across the three drift conditions, whereas both smaller and larger gains result in higher ACSL. These results highlight the importance of appropriately regulating the magnitude of fulfillment-driven updates.

Fig.~\ref{fig:control_trajectories} compares the adaptive-control trajectories of Dynamic STEPS and Frozen-After-Drift under the same pre-drift history. After the drift event at $t_{\mathsf d}=50$, Frozen-After-Drift retains its control parameters and resource prices, whereas Dynamic STEPS continues updating them according to execution feedback. As shown in Figs.~\ref{fig:control_trajectories}(a)--(d), Dynamic STEPS reduces its semantic admission threshold, while the contract conservativeness factor and edge coordination gain initially increase before gradually decreasing. Meanwhile, the bandwidth resource price rises and subsequently stabilizes. These adjustments reflect the joint effects of fulfillment-driven control updates and congestion-sensitive resource pricing under changing service conditions.

As reported in Table~\ref{tab:post-drift-feedback}, Dynamic STEPS achieves higher AFD, OAFD, and SAR than Frozen-After-Drift, while reducing CVR and ACSL. Together with the control trajectories in Fig.~\ref{fig:control_trajectories}, these results demonstrate that continued fulfillment-driven adaptation improves contract fulfillment and service coverage under the tested joint drift.

\begin{figure*}[!t]
	\centering\color{black}
	\includegraphics[width=0.98\textwidth]{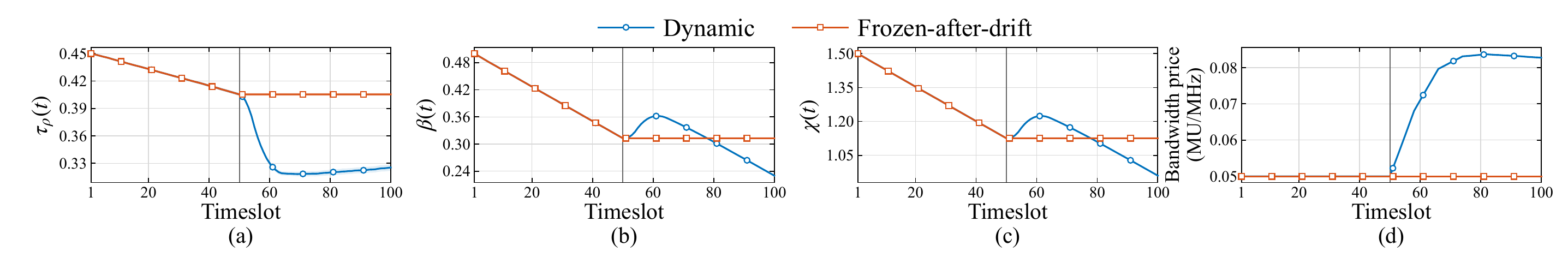}
	\caption{Adaptive-control trajectories under joint drift: (a) semantic admission threshold, (b) contract conservativeness factor, (c) edge coordination gain, and (d) bandwidth resource price.}
	\label{fig:control_trajectories}
\end{figure*}

\subsection{Baseline Implementation and Runtime}

\begin{table}[!t]
	\centering\color{black}
	\caption{Baseline implementation and runtime.}
	\label{tab:table_c4}
	\footnotesize
	\setlength{\tabcolsep}{3pt}
	\renewcommand{\arraystretch}{1.15}
	
	\begin{tabular}{@{}p{0.29\columnwidth}p{0.63\columnwidth}@{}}
		\toprule
		\multicolumn{2}{@{}l}{\textbf{(a) Adapted benchmark scope}} \\
		\midrule
		\textbf{Method} & \textbf{Mechanism} \\
		\midrule
		Param-Opt & Reference numerical requirements; potential-based best response; no semantic admission \\
		Direct-LLM & Shared semantic parsing outputs; direct contract construction; potential-based best response \\
		DRL-Opt & Fixed Double-DQN policy; feasible-action masking; validation-selected checkpoint \\
		Intent-Reconf & Violation-triggered correction on a subsequent request \\
		GAH & Urgency-ordered, capacity-aware greedy service-loss minimization \\
		\bottomrule
	\end{tabular}
	
	\vspace{2mm}
	
	\begin{tabular}{@{}p{0.29\columnwidth}p{0.63\columnwidth}@{}}
		\toprule
		\multicolumn{2}{@{}l}{\textbf{(b) Double-DQN training protocol}} \\
		\midrule
		\textbf{Item} & \textbf{Setting} \\
		\midrule
		Architecture & Two $128$-unit hidden layers with ReLU activation; Huber loss; discount factor $0.95$ \\
		Optimizer & Adam; learning rate $10^{-3}$; batch size $128$; gradient clipping threshold $10$ \\
		Replay & Capacity $100,000$; warm-up $2,000$ training steps \\
		Target update & Every $500$ gradient steps \\
		Training & Up to $60,000$ steps; random seeds $\{1,2,3\}$ \\
		Validation & Every $5,000$ steps; median seed-best selection \\
		Data separation & Independent training, validation, and test environments \\
		\bottomrule
	\end{tabular}
	
	\vspace{2mm}
	
	\begin{tabular*}{\columnwidth}{@{\extracolsep{\fill}}lrr@{}}
		\toprule
		\multicolumn{3}{@{}l}{\textbf{(c) Semantic-processing latency}} \\
		\midrule
		\textbf{Operation} & \textbf{Samples} & \textbf{Mean (s)} \\
		\midrule
		Uncached inference & 480 & 0.905283 \\
		Cache lookup & 480 & 0.000326 \\
		Five-call self-consistency & 360 & 5.321906 \\
		\bottomrule
	\end{tabular*}
	
	\par\smallskip
	\noindent\parbox{\columnwidth}{\footnotesize\textit{\textbf{Note:} Panel (b) reports the training and validation settings of DRL-Opt. Panel (c) reports semantic-processing latency separately from scheduling RT. Cache lookup retrieves stored parsing outputs for identical requests.}}
\end{table}

Table~\ref{tab:table_c4}(a) summarizes the scheduling mechanisms of the five benchmarks under the common experimental setting. These methods differ in their treatment of numerical service requirements, semantic interpretation, learned scheduling policies, violation-triggered reconfiguration, and resource-aware greedy assignment, as detailed in Sec.~\ref{sec:benchmarks}.

Table~\ref{tab:table_c4}(b) presents the Double-DQN training configuration of DRL-Opt. The policy network employs two hidden layers with ReLU activation and is trained using the Huber loss and Adam optimizer. Independent training, validation, and test environments are used, with the validation-selected policy fixed during test evaluation.

Table~\ref{tab:table_c4}(c) reports the semantic-processing latency of Qwen3.5-4B. Uncached inference incurs substantially higher latency than retrieving cached parsing outputs, while five-call self-consistency requires additional inference time. These measurements characterize the computational overhead of semantic interpretation, which is reported separately from the per-slot scheduling RT defined in Sec.~\ref{sec:metrics}.

\subsection{Comparison Fairness and Statistical Significance}

\begin{table}[htb]
	\centering\color{black}
	\caption{Common-served fulfillment and paired comparisons at $r_{\mathsf{unc}}=0.20$.}
	\label{tab:table_c5}
	\footnotesize
	\setlength{\tabcolsep}{1.8pt}
	\renewcommand{\arraystretch}{1.15}
	
	\begin{tabular*}{\columnwidth}{@{\extracolsep{\fill}}lrrr@{}}
		\toprule
		\multicolumn{4}{@{}l}{\textbf{(a) Common-served requests}} \\
		\midrule
		\textbf{Method} & \textbf{AFD}$\uparrow$ & \textbf{CVR}$\downarrow$ & \textbf{ACSL}$\downarrow$ \\
		\midrule
		\textbf{STEPS} & \textbf{0.9438} & \textbf{0.2794} & 0.4758 \\
		Param-Opt & 0.9429 & 0.2858 & \textbf{0.4677} \\
		Direct-LLM & 0.9163 & 0.3405 & 0.5454 \\
		DRL-Opt & 0.9133 & 0.3435 & 0.5391 \\
		Intent-Reconf & 0.8851 & 0.4253 & 0.5835 \\
		GAH & 0.9127 & 0.3518 & 0.5309 \\
		\bottomrule
	\end{tabular*}
	
	\begin{tabular*}{\columnwidth}{@{\extracolsep{\fill}}lrrrrr@{}}
		\toprule
		\multicolumn{6}{@{}l}{\textbf{(b) Paired differences: STEPS $-$ baseline}} \\
		\midrule
		\textbf{Baseline} & $\Delta\mathrm{AFD}$ & $\Delta\mathrm{OAFD}$ & $\Delta\mathrm{CVR}$ & $\Delta\mathrm{ACSL}$ & $\Delta\mathrm{SAR}$ \\
		\midrule
		Param-Opt & 0.0009 & -0.0041 & -0.0064 & 0.0081 & -0.0054 \\
		Direct-LLM & 0.0276 & 0.0225 & -0.0612 & -0.0696 & -0.0054 \\
		DRL-Opt & 0.0305 & 0.0255 & -0.0639 & -0.0632 & -0.0054 \\
		Intent-Reconf & 0.0587 & 0.0537 & -0.1459 & -0.1077 & -0.0054 \\
		GAH & 0.0312 & 0.0261 & -0.0724 & -0.0551 & -0.0054 \\
		\bottomrule
	\end{tabular*}
	
	\par\smallskip
	\noindent\parbox{\columnwidth}{\footnotesize\textit{\textbf{Note:} Results are averaged over $30$ paired trials. Panel (a) evaluates requests served by all six methods, while panel (b) uses each method's original evaluation population. Differences in panel (b) are calculated as STEPS minus the corresponding baseline. Paired comparisons are assessed using the Wilcoxon signed-rank test with Holm correction~\cite{frank-afli-2026-pteb}.}}
\end{table}

To compare fulfillment performance on the same served-request population, Table~\ref{tab:table_c5}(a) reports AFD, CVR, and ACSL for the intersection of requests served by all six methods. Under this common-served evaluation, STEPS achieves the highest AFD and lowest CVR, while Param-Opt attains slightly lower ACSL. Compared with Direct-LLM, DRL-Opt, Intent-Reconf, and GAH, STEPS maintains higher fulfillment and lower service loss across all three metrics. These results indicate that its fulfillment advantage over the four practical benchmarks persists when evaluated on the same requests.

Table~\ref{tab:table_c5}(b) further reports paired performance differences using each method's original evaluation population. Compared with the four practical benchmarks, STEPS achieves higher AFD and OAFD, together with lower CVR and ACSL. Against Param-Opt, STEPS improves AFD by $0.0009$ and reduces CVR by $0.0064$, whereas Param-Opt achieves higher OAFD and lower ACSL. The difference in overall fulfillment is consistent with the lower SAR of STEPS, as unserved requests contribute zero fulfillment.

After Holm correction, the AFD difference between STEPS and Param-Opt is not statistically significant (adjusted $p=0.056$), whereas the CVR reduction is significant (adjusted $p=0.0024$). Param-Opt achieves significantly lower ACSL (adjusted $p<0.001$), while the fulfillment improvements of STEPS over the four practical benchmarks are significant across AFD, OAFD, CVR, and ACSL (adjusted $p<0.001$). Overall, STEPS achieves competitive performance against the parameter-driven benchmark while providing improved contract fulfillment over the four practical alternatives.
}

\end{document}